\documentclass[11pt]{article}
\usepackage{amsmath,amsfonts,amsthm,amssymb, graphicx}
\swapnumbers
\theoremstyle{plain}

\newtheorem{thm}{THEOREM}[section]
\newtheorem{lm}[thm]{LEMMA}

\theoremstyle{definition}
\newtheorem{defi}[thm]{DEFINITION}
\theoremstyle{definition}

\newcommand{\upchi}{\raise1pt\hbox{$\chi$}}
\newcommand{\R}{{\mathord{\mathbb R}}}

\newcommand{\C}{{\mathord{\mathbb C}}}
\newcommand{\Z}{{\mathord{\mathbb Z}}}
\newcommand{\N}{{\mathord{\mathbb N}}}

\newcommand{\E}{{\mathcal{E}}}

\newcommand{\tr}{{\rm Tr}}
\renewcommand{\|}{{\Vert}}
\numberwithin{equation}{section}
\def\dd{{\rm d}}

\def\mc{{\mu_{{\rm C}}}}
\def\ha{{{\rm Hol},\alpha}}

\def\T{\mathbb{T}^1}
\def\H{\mathcal{H}}

\def\K{\mathcal{K}}
\def\Q{\mathcal{Q}}
\def\P{\mathcal{P}}
\def\olp{\overline{\phi}}
\begin{document}

\title{Exponential Relaxation to Equilibrium for a One-Dimensional
Focusing Non-Linear Schr\"odinger Equation with Noise.}

\author{Eric A. Carlen, J\"urg Fr\"ohlich and Joel Lebowitz}

\maketitle

\abstract

We construct generalized grand-canonical- and canonical Gibbs measures for
a Hamiltonian system described in terms of a complex scalar field that is defined 
on a circle and satisfies a nonlinear Schr\"odinger equation with a focusing 
nonlinearity of order $p<6$. Key properties of these Gibbs measures, in particular
 absence of ``phase transitions'' and regularity properties of field samples, are established.
We then study a time evolution of this system given by the Hamiltonian evolution 
perturbed by a stochastic noise term that mimics effects of coupling the system
to a heat bath at some fixed temperature. The noise is of Ornstein-Uhlenbeck type for
the Fourier modes of the field, with the strength of the noise decaying to zero, as the 
frequency of the mode tends to $\infty$.
We prove exponential approach of the state of the system to a grand-canonical Gibbs 
measure at a temperature and ``chemical potential'' determined by the stochastic noise term.

\section{Introduction}

The nonlinear Schr\"odinger equation (NLS) governs the time evolution of a system 
described in terms of a complex scalar field,  $\phi(x,t)$, where $x$ is a point in an 
open domain $U\subset \R^d$, and $t\in \mathbb{R}$ denotes time. The NLS equation is a Hamiltonian evolution equation corresponding to a Hamilton functional (Hamiltonian)
\begin{equation}\label{ham0}
H(\phi) = \frac12 \int_U |\nabla \phi(x)|^2\dd x  -\frac{\lambda}{p} \int_U |\phi(x)|^p \dd x  ,
\end{equation} 
$\lambda\in\mathbb{R}$, defined on a suitably chosen affine phase space on which it generates a canonical flow. 

The NLS equation has been used to model a broad variety of physical phenomena ranging from Langmuir waves  in plasmas to signal propagation in optical fibers \cite{SS, KMM}. 

It is of interest to couple the system with Hamiltonian (\ref{ham0}) to a heat bath in thermal equilibrium at an inverse temperature $\beta < \infty$ and to study the relaxation of the system to an equilibrium state at the temperature of the heat bath. As a prerequisite, we must first study the system with Hamilton functional given by (\ref{ham0}) from the point of view of the equilibrium statistical mechanics of fields and then analyze the effects of thermal noise generated by the heat bath on the dynamics of the system. The equilibrium Gibbs measures associated to (\ref{ham0}) are of special interest, because, formally, they are invariant under the Hamiltonian flow generated by (\ref{ham0}).
The construction of such Gibbs measures and the extension of the Hamiltonian flow to the non-smooth fields $\phi$ in the support of these measures
are particularly interesting in the {\em one-dimensional focusing case}, corresponding to $\lambda > 0$, with which this paper is concerned. (It is not possible to give mathematical meaning to the formal Gibbs measures for the \textit{focusing} NLS, with a non-trivial dependence of expectations on the coupling constant $\lambda>0$, in dimension $d\geq 2$; see \cite{Brydges}.) 
For $d=1$ and $\lambda \leq 0$, the Hamiltonian $H(\phi)$ is bounded from below, and it is not difficult, at all, to construct the corresponding Gibbs measures and to show that they are supported on H\"older continuous fields. Borrowing results from constructive quantum field theory \cite{GJ}, one can actually show that, in the \textit{defocusing} case, and for $p=4$, Gibbs measures can also be constructed in two and three dimensions.

The situation is very different in the \textit{focusing} case, because $H(\phi)$ is not bounded below,
 for $\lambda > 0$, and, worse, $e^{-\beta H(\phi)}$ is not integrable with respect to the formal ``Liouville measure'' $\mathcal{D}\phi\mathcal{D}\overline{\phi}$. (The precise meaning of this statement will be made clear below.) However, using that the $L^{2}$- norm of $\phi$ is a conserved quantity, one can actually construct what might be called  ``grand-canonical'' Gibbs measures for all
$\lambda >0$, assuming that  $p< 6$, 
and for $p = 6$, provided $\lambda$ is sufficiently small; as has first been outlined in \cite{LRS}.  
The existence of a Hamiltonian flow acting on the support of these measures was first proven 
by Bourgain in \cite{B}.

Ergodicity properties of the Hamiltonian flow are of obvious interest, but are hard to explore, mathematically.
However, if  the Hamiltonian flow is perturbed by  thermal noise emanating from a heat bath, which we model by a suitable source of random noise with the property that the resulting stochastic flow leaves the Gibbs measure corresponding to the temperature of the heat bath invariant, one can hope to establish exponential relaxation of a large class of initial states of the system towards
that Gibbs measure. For a finite-dimensional approximation of the system described by (\ref{ham0}) -- a truncation of Fourier modes -- this has been accomplished in \cite{MCL} and \cite{LW}. 

In this paper we present a new construction of the generalized ``grand-canonical'' Gibbs measures of \cite{LRS} and, moreover, of a ``canonical'' Gibbs measure, as explained below. We then introduce and study a stochastic evolution of the system leaving a generalized ``grand-canonical'' Gibbs measure
invariant, and we exhibit exponential convergence of a large class of initial conditions to this measure. 
We do $not$ need to introduce any truncation or finite-dimensional approximation of our system. 
We also discuss the corresponding problem for the ``canonical'' Gibbs measure, but do not treat it in full detail in this paper. 
Finally, for the ``grand-canonical'' Gibbs measures, we prove \textit{absence of phase transitions} under variations of the parameters $\beta$ and $\lambda$ and a generalized chemical potential (introduced below).  
Related questions concerning he evolution of infnite dimensional measures has been investigated under different evolutions, e.g., the Hartree flow, in \cite{AN1,AN2}. 
It is of interestest to develop a toolkit that is relevant to a wide class of such problems, and we hope to make a contribution to this here. 

In the next section, we give a precise formulation of the mathematical problems to be reckoned with, in this paper.

\section{Precise formulation of the problems}

Consider the non-linear Schr\"odinger equation,
$$
i\frac{\partial }{\partial t} \phi  = -\Delta \phi  - \lambda |\phi|^{p-2}\phi,
$$ 
on the circle $\T$  of circumference $L$, with $\lambda> 0$. 

\smallskip
This equation is the Hamiltonian equation of motion corresponding to the Hamilton functional 
\begin{equation}\label{ham}
H(\phi) = \frac12  \int_{\T} |\phi'(x)|^2{\rm d}x  - \frac{\lambda} {p} \int_{\T} |\phi(x)|^p {\rm d}x\ 
\end{equation}
defined on the complex Sobolev space $H^1(\T)$, which can be interpreted as an affine phase space.
Since every function $\phi \in H^1(\T)$ is bounded, and hence in $L^p$ for all $p$, the Hamiltonian is well-defined and finite on all of $H^1(\T)$.
The phase space for this system is the real Hilbert space $\K$ obtained by regarding complex $H^{1}(\T)$ as a real Hilbert space and equipping it with the inner product 
\begin{equation}\label{inner}
\langle \phi,\psi\rangle_\K = \Re\left(\langle \phi,\psi\rangle_{L^2}\right)\ 
\end{equation}
where $\Re(z) $ is the the real part of $z\in \C$. (Note that $\K$ is not complete with respect to this
inner product; it is complete with respect to the inner product given by
$$
 \Re\left(\langle \phi, (-\Delta + 1) \psi\rangle_{L^2}\right)\ 
 $$
 But this does not present a problem. For many purposes, it is actually more natural to define phase space as
 a space of H\"older-continuous functions on the circle $\mathbb{T}^1$ equipped with the quadratic form defined in \eqref{inner}.)

In this infinite-dimensional setting, the Liouville measure does not exist, since, formally, it would be the Lebesgue measure on $\K$. However, Gibbs measures can be constructed on a function space properly 
containing $\K$.
Using a standard physics notation, $\mathcal{D} \phi  \mathcal{D} \overline{\phi}$, for the formal Liouville measure,
the {\em formal} Gibbs measure corresponding to (\ref{ham}) would be 
\begin{equation}\label{fg}
\frac{1}{Z} e^{-\beta H(\phi)} \mathcal{D} \phi  \mathcal{D} \overline{\phi}\ ,
\end{equation}
where $Z$ is a normalization constant.  Considered as a function of the parameters $\beta$, $\lambda$ and $L$,
$Z$  is  the  {\em partition function}.  To make sense of (\ref{fg}), one must show that $Z$ is finite; (i.e., that 
$e^{-\beta H(\phi)}$ is ``integrable'' with respect to $\mathcal{D} \phi  \mathcal{D} \overline{\phi}$). 
 
A key first step towards making sense of (\ref{fg}) is -- as is well known --  to define (\ref{fg}) for 
$\lambda = 0$, in which case only the kinetic energy term is present in the Hamiltonian. 
Expression (\ref{fg}) then defines a Gaussian measure, namely the Wiener measure on the space of periodic Brownian paths with period $L$, and
\begin{equation}\label{formg}
\dd \mu_\beta = \frac{1}{Z_0} e^{-(\beta/2) [\|\phi'\|_2^2 + m^2\|\phi\|_2^2]} \mathcal{D} \phi  \mathcal{D} \overline{\phi}\ ,
\end{equation}
with $m^{2} > 0$, is a formal expression for the Gaussian probability measure, $\mu_\beta$, with mean zero and covariance given by
$$
C := \beta^{-1}\left(m^2-\Delta\right)^{-1}\ ,
$$
where $\Delta$ denotes the Laplacian on $\T$, which is a well-defined probability measure 
on the space, $\Omega$, of (H\"older-) continuous functions on $\T$. The second moment (covariance) of $\mu_{\beta}$ is given by
$$\int_\Omega |\langle h,\phi\rangle_{L^2}|^2 \dd \mc(\phi) = \langle h,C h \rangle_{L^2}\ .$$

The Gaussian measure $\mu_\beta$ is the natural ``reference measure'' for our system, and we shall use it to construct a normalized ``grand-canonical'' Gibbs measure.
We will then introduce  stochastic perturbations of the Hamiltonian flow determined by (\ref{ham}) modeling the effects of thermal noise, for which the Gibbs measures constructed in this paper are the unique invariant measures. We will show that the corresponding stochastic dynamics drives a general class of initial states of the system exponentially fast to such a Gibbs measure. 

All this is  explained in more detail in subsequent sections. We start by recalling some well-known facts concerning the 
non-linear Schr\"odinger equation from the point of view of classical Hamiltonian systems with infinitely many degrees of freedom.

\subsection{The Hamiltonian flow}  

Let the Hamiltonian $H$ be defined by (\ref{ham}) on the  real Hilbert space $\K$, with 
inner product defined in (\ref{inner}).  
Thus, for any non-zero $\phi \in H^{1}(\T)$,
$\phi$ and $i\phi$ are linearly independent in $\K$.  A symplectic structure on $\K$ is introduced as follows:
Define a complex structure, $J$, on $\K$ to be the orthogonal transformation on the space, $T\K$, of vector fields on $\mathcal{K}$ given by
\begin{equation}\label{Jdef}
J\phi = -i\phi\ .
\end{equation}
Evidently $J^2= -I$.  The symplectic 2-form, $\omega$, is then defined by
$$
\omega(\phi,\psi) = \langle \phi,J\psi\rangle_\K\ ,
$$
for arbitrary $\phi$ and $\psi$ in $T\K$. (Using the inner product (\ref{inner}), one-forms on $\K$ and vector fields on $\K$ will henceforth be identified with functions in $\K$, without any danger of confusion.)

For any sufficiently regular, real-valued function $F$ on $\K$, we denote the Fr\'echet derivative of $F$ by $DF$. In other words, $DF(\phi)$ is the unique element of $\K$
such that 
\begin{equation}\label{Fder}
\langle DF(\phi),\psi\rangle_\K  = \frac{{\rm d}}{{\rm d}t}F(\phi+ t\psi)\bigg|_{t=0} \ ,
\end{equation}
for arbitrary $\psi\in \K$.
Associated to any Fr\'echet-differentiable function, $F$, on $\K$ there is a {\em Hamiltonian vector field} on $\K$ given by
$$X_F(\phi) = JDF(\phi)\ .$$
Suppose that the ODE 
$$
\frac{{\rm d}}{{\rm d}s}\phi(s)  = X_F(\phi(s))\ ,\qquad \phi(0) = \phi_0,
$$
has a unique solution, for all $\phi_{0}\in \K$.  Let  $\Phi^s_F$ be the associated {\em flow} on $\K$ given by $\Phi^s_F(\phi_0) = \phi(s)$.
Then the transformations $\{\Phi_F^s\}_{s\in \R}$ form a one-parameter group under composition.   

As an example, we choose $F$ to be given by $F(\phi) = \langle \phi, \phi \rangle_{\K}$, which is also the squared $L^2(\T)$- norm  of $\phi$.
One readily verifies that  $DF = 2\phi$, and hence $X_F(\phi) = -2i\phi$. The Hamiltonian flow 
associated to $F$, i.e., the flow given by solving the equation
$$
\frac{{\rm d}}{{\rm d}s}\phi(s) = X_F(\phi) = -2i\phi, \qquad \phi(0) = \phi_0,
$$
is simply $\phi(s,x) = e^{-2is}\phi_0(x)$, a simple phase transformation (``gauge transformation of the first kind''). 

The Hamiltonian vector field $X_H(\phi)$ associated to the Hamilton functional $H$ is not defined on all of 
$\K$, but only on the dense subset $H^2(\T)$. In fact, $H$ is Fr\'echet-differentiable at $\phi$ if and only if
 $\phi\in H^2(\T)$, in which case
 $$
 DH(\phi) = -\Delta \phi -\lambda |\phi|^{p-2}\phi\ .
 $$
 Thus, for $\phi\in H^2(\T)$, we have that
 $X_H(\phi) = -i\left(-\Delta \phi -\lambda |\phi|^{p-2}\phi\right)$,
and the Hamiltonian flow $\Phi_H^t$ generated by this vector field is given by setting $\Phi_H^t(\phi_0)$ equal to $\phi(t)$, the solution of the NLS equation
 $$
 \frac{\rm{d}}{\rm{d}t}\phi(t) = X_H(\phi) = i\Delta \phi + i\lambda |\phi|^{p-2}\phi, \qquad \phi(0) = \phi_0,
 $$
 at time $t$. (We use $t$, for time, as the parameter of the Hamiltonian flow, instead of the usual $s$.) 
 Writing this out as a PDE on $\T$, instead of as an ODE in $\K$, we have the  non-linear Schr\"odinger equation
 on $\T$:
\begin{equation}\label{nls}
i\frac{\partial}{\partial t} \phi  = -\Delta \phi  - \lambda |\phi|^{p-2}\phi\ .
\end{equation} 

The {\em Poisson bracket} on pairs of functions on phase space $\K$ is given by 
$$\{F,G\} = \langle DF, J(DG)\rangle_\K\ .$$
Since $J$ is antisymmetric, $\{F,G\} = -\{G,F\}$.  By definition 
$$\{F,G\}(\phi)  =  DF(X_G)(\phi)  =  \frac{{\rm d}}{{\rm d}s} F( \Phi_{G}^s\phi)\bigg|_{s=0}\ .$$
Since the Poisson bracket is anti-symmetric, we also have that
$$\{F,G\} (\phi)=  - DG( X_F)(\phi)  =  \frac{{\rm d}}{{\rm d}s} G( \Phi_{F}^s\phi)\bigg|_{s=0}\ .$$

In particular, taking $G = H$, we see that if  the Hamiltonian $H$ is invariant under the flow $\Phi_F^s$,
so that 
$\frac{{\rm d}}{{\rm d}s} H( \Phi_{F}^s\phi) = 0$,  then $\{F,H\} = 0$, and so  $\frac{{\rm d}}{{\rm d}t} F( \Phi_{H}^t\phi) = 0$.
That is, 
$$
\frac{{\rm d}}{{\rm d}s} H( \Phi_{F}^s\phi) = 0 \quad \Rightarrow \quad   \frac{{\rm d}}{{\rm d}t} F( \Phi_{H}^t\phi) = 0\ ,
$$
i.e., any F that generates a dynamical symmetry is a conservation law.
In particular, since $H(e^{i2s}\phi) = H(\phi)$, for all $s$ and $\phi$, the function $F(\phi) = \|\phi\|_\K^2$ generates a dynamical symmetry of the system (i.e.,  $\frac{{\rm d}}{{\rm d}s} H( \Phi_{F}^s\phi) = 0$), and hence $\|\phi(t)\|_\K^2 = \|\phi\|_{L^2(\T)}^{2}$ is a conservation law. In other words, $\|\phi(t)\|_{L^2(\T)}^{2}$ is constant along the flow generated by $H$; i.e., along solutions of the non-linear Schr\"odinger equation.  This fact can of course be seen directly, but since it plays an important role in what follows it is useful to recall that the conservation law for the $L^2$ norm is a consequence of the invariance of the Hamiltonian under phase transformations.

\if false

\subsection{Darboux coordinates}

We shall need coordinates to do computations and make approximations.  Let $\Q$ be the subspace of $\K$ consisting
of real valued functions, and let $\P$ be the subspace consisting of imaginary valued functions. Then these two subspaces
are orthogonal complements of one another; i.e., $\K = \Q \oplus \P$, and $J$ is an orthogonal transformation of 
$\Q$ onto $\P$, and hence $\P$ onto $\Q$, since $J^2 = -I$. 

Let 
 $\{u_n\}$ be  any orthonormal basis of $\Q$.   Let $v_n= -Ju_n$. Then $\{v_n\}$ is an orthonormal basis for $\P$. 
 Let $P_\Q$ and $P_\P$ denote the orthogonal projections onto $\Q$ and $\P$ respectively. 
 Note that  both  $\{u_n\}$ and $\{v_n\}$ are  also  orthonormal bases of $L^2(\T)$, 
 and hence
\begin{equation}\label{del}
\sum_{n=1}^\infty u_n(x)u_n(y) =  \delta(x-y) \ .
\end{equation}  
However, for any $\phi\in H^1(\T)$, so that $\phi(x)$ is well-defined  for each $x$,
 \begin{equation}\label{can9}
 \sum_{j=1}^{n} \langle \phi ,u_n\rangle_\K u_n(x) = P_\Q \phi(x) \qquad{\rm and}\qquad 
\sum_{j=1}^{n} \langle \phi ,v_n\rangle_\K v_n(x) = P_\P \phi(x)\ .
 \end{equation}

 For any $\phi\in \K$, define
 $$q_n(\phi) = \langle \Re \phi,u_n\rangle_\K \qquad{\rm and}\qquad p_n(\phi) = \langle \Im \phi,u_n\rangle_\K\ ,$$
so that 
 \begin{equation}\label{can10}
 \phi   = \sum_{n=1}^\infty (q_n(\phi)+ip_n(\phi))u_n \qquad{\rm and}\qquad  \olp = \sum_{n=1}^\infty (q_n(\phi)-ip_n(\phi))u_n\ .
 \end{equation}

 Furthermore, since $\Im \phi$ is the projection of $i\phi$ onto $\Q$,  $\langle \Im \phi,u_n\rangle_\K = \langle -i\phi,u_n\rangle_\K = \langle  \phi,-Ju_n\rangle_\K$, and therefore
  \begin{equation}\label{can11}
  Dq_n(\phi) = u_n \qquad{\rm and}\qquad Dp_n(\phi) = -Ju_n = v_n\ .
   \end{equation}
 independent of $\phi$, and so, for all $m$ and $n$, 
 \begin{equation}\label{can2}
 \{q_m,q_m\}  = \{p_m,p_n\} = 0\qquad{\rm and}\qquad \{q_m,p_n\} =  \delta_{m,n}\ .
 \end{equation}
 
Also, since when $D$ is applied to complex valued functions it is separately  applied to the real and imaginary parts,
$$
D\left( \sum_{n=1}^m (q_n(\phi)+ip_n(\phi))u_n(x)\right) = \sum_{n=1}^m [u_n u_n(x) + v_n v_n(x)]\ .$$

For $\psi\in H^2(\T)$,  not only is the evaluation functional $x\mapsto \phi(x)$  well defined, but  
$$\sum_{n=1}^\infty  [\langle \psi, u_n\rangle_\K  u_n(x) + \langle \psi, v_n\rangle_\K v_n(x)] = \psi(x)$$
where the series converge absolutely  and where  
where (\ref{can9}) identifies the sum. 
By (\ref{can10}), 
$\phi(x) = \lim_{m\to \infty}\sum_{n=1}^m (q_n(\phi)+ip_n(\phi))u_n(x)$. 
Combining the last two results, and letting $[\phi(x)]$ denote the functional that evaluates $\phi$ at $x$, we have 
that everywhere on $H^1(\T)$,
\begin{equation}\label{can12}
D[\phi(x)] = \delta(x-y)\ .
\end{equation}

There is a useful and more traditional way to write this. Taking $p$ and $q$ as the usual real coordinates on the complex plane, 
let $\phi = q+ip$ and $\olp = q-ip$. Then a simple computation shows that
$$\frac{\partial}{\partial \phi} = \frac12\left( \frac{\partial}{\partial q} - i \frac{\partial}{\partial p}\right) \qquad{\rm and}\qquad
\frac{\partial}{\partial \olp} = \frac12\left( \frac{\partial}{\partial q} + i \frac{\partial}{\partial p}\right)\ .$$
Of course, we then have $\partial \phi/\partial \phi = \partial \olp/\partial \olp =1$ and 
$\partial \phi/\partial \olp = \partial \olp/\partial \phi =0$.

For $F$ that is Fr\'echet differentiable we define
 \begin{equation}\label{can3}
 \frac{\delta F}{\delta \phi} = \frac12 \left( P_\Q DF   + P_\P DF\right)\qquad{\rm and}\qquad 
  \frac{\delta F}{\delta \olp}  = \frac12 \left( P_\Q DF   -  P_\P DF\right)\ ,
 \end{equation}
 and
  \begin{equation}\label{can3C}
 \frac{\delta F}{\delta q} =  P_\Q DF  \qquad{\rm and}\qquad 
  \frac{\delta F}{\delta p}  = JP_\P DF\ .
 \end{equation}
 Recalling that $J$ is multiplication by $-i$, and that for $F$ real valued, $P_\Q F$ is a real valued function, and 
 $P_\P DF$ is imaginary valued, so that $\delta F/\delta p$ is real valued. 
 We note that
  \begin{equation}\label{can3D}
 \frac{\delta }{\delta q} =   \left(\frac{\delta }{\delta \phi} + \frac{\delta }{\delta \olp} \right)  \qquad{\rm and}\qquad 
 \frac{\delta }{\delta p} =   J\left(\frac{\delta }{\delta \phi} - \frac{\delta }{\delta \olp} \right) \ .
 \end{equation}
 We now compute $\{F,G\} = \langle DF,JDG\rangle_\K = \langle P_\Q DF,JP_\P DG\rangle_\K + \langle P_\P DF,JP_\Q DG\rangle_\K$ in these terms:
 $$\{F,G\} = \int_{\T}  \left[  \frac{\delta F}{\delta q} \frac{\delta G}{\delta p}-   \frac{\delta F}{\delta p} \frac{\delta G}{\delta q}\right]{\rm d}x $$
 and then. by (\ref{can3D}), 
 $$\{F,G\} = 2i\int_{\T}  \left[  \frac{\delta F}{\delta \phi} \frac{\delta G}{\delta \olp}-   \frac{\delta F}{\delta \olp} \frac{\delta G}{\delta \phi}\right]{\rm d}x $$

Finally, a slightly formal but suggestive way of rewriting   (\ref{can12})  is:
\begin{equation}\label{can13}
\frac{\delta \phi(x)} {\delta \phi(y)} =  \frac{\delta \olp(x)} {\delta \olp (y)} =  \delta(x-y)
\qquad{\rm and}\qquad \frac{\delta \phi(x)} {\delta \olp(y)} = \frac{\delta \olp(x)} {\delta \phi(y)} = 0\ .
\end{equation}

We now compute 
$\{\phi(x),H\}$ which gives the evolution of $\phi(x)$ under the Hamiltonian flow using the coordinates introduced above.
First, 
$$\frac{\delta}{\delta \olp} H = -\frac12 \Delta \phi -\frac12\lambda |\phi|^{p-2}\phi\ ,$$
By (\ref{can13}), 
$$
\{\phi(x),H\} =  i \int_{\T} \delta(x-y) \left[ -\Delta \phi(y) -\lambda |\phi|^{p=2}\phi(y)\right]{\rm d}y 
= i[-\Delta \phi - \lambda |\phi|^{p-2}\phi](x)\ ,
$$
and so  the Hamiltonian flow acts on $\phi(x)$ so that $\phi(x,t)$ solves the non-linear Shroedinger equation 
\begin{equation}\label{nls}
i\frac{\partial }{\partial t} \phi  = -\Delta \phi  - \lambda |\phi|^{p-2}\phi\ .
\end{equation}

\fi

\subsection{Gibbs measures on phase space} 

Recalling the formal expression (\ref{formg}) for the Gaussian reference measure $\mu_\beta$, the  formal canonical 
Gibbs measure for our system is  given by 
\begin{eqnarray}\label{canX}
\frac{1}{Z} e^{-\beta H(\phi)} \mathcal{D} \phi  \mathcal{D} \overline{\phi}  &= &
\frac{Z_0}{Z} e^{\beta\frac{\lambda}{p} \|\phi\|_p^p + \frac{\beta}{2} m^2\|\phi\|_2^2}
\frac{1}{Z_0} e^{-\frac{\beta}{2} [\|\phi'\|_2^2 + m^2\|\phi\|_2^2]} \mathcal{D} \phi  \mathcal{D} \overline{\phi} \nonumber\\
&=&  \frac{Z_0}{Z} e^{\beta\frac{\lambda}{p} \|\phi\|_p^p + \frac{\beta}{2} m^2\|\phi\|_2^2}  \dd \mu_\beta\ .
\end{eqnarray}
For $\lambda \leq 0$ and $p>2$, there is no difficulty in proving that  the exponential of $\beta\frac{\lambda}{p} \|\phi\|_p^p + \beta m^2\|\phi\|_2^2$
is integrable with respect to $\mu_\beta$. Thus, in this case, one obtains a well defined normalized Gibbs measure almost without effort. This measure is related to the path-space measure appearing in the Feynman-Kac formula for the (imaginary-time) propagator of the anharmonic oscillator.

The situation is quite different when $\lambda > 0$. Then, for $p>2$, the exponential of $\beta\frac{\lambda}{p} \|\phi\|_p^p + \beta m^2\|\phi\|_2^2$
is \textit{not} integrable with respect to $\mu_\beta$.   However,  for each $n>0$, the spherical subset of $\K$ defined by
$$\{ \phi \ :\ \|\phi\|_2^2 = n\}$$
is invariant under the Hamiltonian flow.   We introduce the random variable
\begin{equation}\label{Ndef}
N(\phi) = \|\phi\|_2^2\ .
\end{equation}
Below we will show that, with respect to $\mu_\beta$, $N(\phi)$ has a smooth density, $\rho_N(s)$, that is strictly positive  on $\R_+$. We then consider the probability measures $\nu_{n, \beta}$ formally defined by 
\begin{equation}\label{nunbf}
\nu_{n,\beta} = \frac{1}{\rho_N(n)}\delta(N(\phi) - n) \mu_\beta\ , \qquad n\in \mathbb{R}_{+}.
\end{equation}
More precisely, $\nu_{n,\beta}$ is the law of $\phi$ under $\mu_\beta$, conditioned on the event that $\{N(\phi) = n\}$. This actually involves conditioning on
a set of measure zero, and since the measure $\nu_{n,\beta}$, while a natural object, is not usually studied in the literature, one task to be carried out below is to show that $\nu_{n,\beta}$ is well defined and to prove certain exponential integrability estimates with respect to 
$\nu_{n,\beta}$.  (The measures studied in \cite{LRS} are defined on the balls $\lbrace \phi \vert N(\phi) \leq n\rbrace \subset \K$.)

Exponential integrability estimates are of interest, because a natural {\em canonical Gibbs measure} $\gamma_{n,\beta,\lambda}$ can then be defined, formally, by 
\begin{equation}\label{microcanX}
\dd \gamma_{n,\beta,\lambda} = \frac{1}{Z(\beta, \lambda, n)}e^{\beta\left[\frac{\lambda}{p} \|\phi\|_p^p + \frac{1}{2}m^{2}n\right]}\dd\nu_{n,\beta}\ ,
\end{equation}
provided that 
$$
Z(n,\beta, \lambda) = e^{\beta m^{2} n}\int_\Omega e^{\beta\frac{\lambda}{p} \|\phi\|_p^p}\dd \nu_{n,\beta} < \infty\ .
$$
We will prove this for all $\lambda > 0$, assuming that $p < 6$.

We shall also be interested in a {\em generalized grand-canonical Gibbs measure}.  
For any $r>0$, we consider the modified Hamiltonian $\tilde H(\phi)$
given by
$$
\tilde H(\phi) = H(\phi) + \kappa N(\phi)^r\ ,
$$
where $N(\phi)$ is the random variable introduced in (\ref{Ndef}).
The  flow generated by this modified Hamiltonian is given by solving the equation
\begin{equation}\label{nlsm}
i\frac{\partial }{\partial t} \phi  = -\Delta \phi  - \lambda |\phi|^{p-2}\phi  + 2r\kappa N(\phi)^{r-1} \phi\ ,
\end{equation} 
where $N(\phi(t)) = N(\phi(0))$ is a constant of the motion, so that the effect of the modification of the Hamiltonian is nothing more than a phase transformation in  the flow.   We shall see that, for 
$\kappa>0$ and $r$ sufficiently large (depending on $p< 6$), 
$\beta\frac{\lambda}{p} \|\phi\|_p^p + \beta m^2\|\phi\|_2^2 - \beta \kappa N(\phi)^r$ is exponentially integrable with respect to 
$\mu_\beta$, and hence
$$
\tilde Z(\kappa,\beta,\lambda) = \int_\Omega e^{\beta\frac{\lambda}{p} \|\phi\|_p^p  - \beta \kappa N(\phi)^r}\dd \mu_{\beta} < \infty\ .
$$

We may think of $N(\phi)$ as a sort of ``particle number observable'' and, consequently, of $\kappa$
 as a sort of ``chemical potential''.   We  define a {\em generalized grand-canonical Gibbs measure}, 
 $\widetilde \gamma_{\beta,\lambda,\kappa}$, by 
\begin{equation}\label{grandcan}
\dd \widetilde \gamma_{\kappa,\beta,\lambda} = \frac{1}{\widetilde Z(\kappa,\beta, \lambda)}e^{\beta\frac{\lambda}{p} \|\phi\|_p^p - \beta \kappa N(\phi)^r}\dd \mu_{\beta} \ .
\end{equation}

We will investigate analyticity properties of the partition functions $Z(n,\beta,\lambda)$ and
 $\widetilde Z(\kappa,\beta,\lambda)$ in the parameters $\beta, \lambda$ and $n$ or $\kappa$, respectively; (see Section 4). These analyticity properties show that the system studied here does not exhibit any phase transition. Our result answers a question posed in \cite{LRS,LRS2}. As described in the next subsection and proven in Section 5, relaxation to the Gibbs measure under a suitably chosen stochastic time evolution is exponential, for an arbitrary choice of the parameters $\beta \in \mathbb{R}$, $\lambda \in \mathbb{R}$ and $\kappa > 0$ compatible with exponential integrability. This is another reflection of the absence of ``phase transitions''.
 
\subsection{Stochastic perturbations of the Hamiltonian flow}

Using the notations introduced in subsection 2.1, the non-linear Schr\"odinger equation can be written as an infinite-dimensional ordinary differential equation:
$$ \dd \phi(t) = JDH(\phi(t))\dd t \ .$$

Let $\sigma$ be a self-adjoint Hilbert-Schmidt operator on $\K$, so that $\sigma^2$ is a positive, trace-class operator on $\K$. Let $w(t)$ denote ``Brownian motion on $\K$'', and consider the stochastic differential equation 
\begin{equation}\label{noise2}
 \dd \phi(t) = JDH(\phi(t))\dd t  - \frac{\beta}{2} \sigma^2 DH(\phi) \dd t + \sigma \dd w(t)\ .
 \end{equation}
 
 We will prove that this equation defines a stochastic process for which the generalized grand-canonical Gibbs measure is invariant, and, moreover, that
 the process drives initial states to this Gibbs measure at an exponential rate. We emphasize that, although we will prove 
 the first claim for all $0<p<6$, $\lambda> 0$ and all $\kappa > 0$, we prove the latter only for $p\leq 4$, and we will not derive an explicit estimate on the rate of exponential relaxation. 
 As will be shown elsewhere, explicit estimates can be obtained, provided 
 $\lambda$ is so small that $H(\phi)$ has a certain convexity property. The required convexity property actually fails for large $\lambda$, c.f. \cite{LW}; but our proof of exponential relaxation still applies. 

In order to prove that (\ref{noise2}) defines a stochastic process that has $ \widetilde \gamma_{\beta,\lambda,\kappa} $ as its inviariant measure, we  require $\sigma^2$ to be trace-class. However, our proof of exponential relaxation requires that $\sigma^2$ not be too small. 
Indeed, if $\sigma^2 = 0$ the stochastic flow reduces to the Hamiltonian flow for which exponential relaxation cannot hold. Existence of solutions of equation (\ref{noise2}) can be proven under broader conditions. For instance, in \cite{DP}, existence of strong solutions of (\ref{noise2}) is proven when 
$\sigma^2$ is the identity.

Under appropriate assumptions on $\sigma^{2}$, the stochastic perturbation of the Hamiltonian flow considered in equation (\ref{noise2}) introduces just the right amount  of  noise to enable us to solve this equation and establish exponential relaxation to a grand-canonical Gibbs measure.  It is possible to turn on stronger noise and, yet, still be able to solve the equation. But it is then not clear how to exhibit exponential relaxation to the grand-canonical Gibbs measure, as we will explain below. In any case, the model 
with $\sigma^2$ being trace-class considered in this paper is quite natural, since  the physical mechanism 
responsible for the noise (namely a coupling of the system to a thermal bath) does not excite all Fourier modes of $\phi$ equally: The trace-class condition can be interpreted as implying that, in some sense, 
only low-frequency modes are excited, as is the case for thermal noise. We do, however, not impose a hard cut-off on the (number of) modes directly coupled to the noise.

Defining a stochastic flow that leaves the canonical Gibbs measures invariant is somewhat 
more subtle and involves introducing noise for which 
the spheres $\{\ \phi\ :\ N(\phi) = n\}$ are invariant subsets. We will discuss such noise in the final section of the paper. Once we will have carried out
our analysis of the stochastic dynamics in the generalized grand-canonical ensemble, it will be easier to explain what needs to be done to treat the canonical ensemble. (A complete analysis of the canonical ensemble is saved for a follow-up paper.)

Our paper is organized as follows. In Section 3, we prove the basic integrability estimates on which our analysis rests and construct the canonical and generalized grand-canonical Gibbs measures. Our construction provides important  ``almost-sure regularity information'' on the functions $\phi$ 
in the support of these measures. 
In Section  4, we prove analyticity properties of the partition functions in the parameters. 
In Section 5, we prove exponential relaxation to a Gibbs measure for the stochastic dynamics. 
In Section 6, we discuss relaxation to a Gibbs measure in the canonical ensemble.

\section{Exponential integrability estimates}

We recall that $\mu_\beta$ is defined to be the Gaussian probability measure with mean zero and covariance
\begin{equation}\label{cov}
C := \beta^{-1}\left(m^2-\Delta\right)^{-1}\ 
\end{equation}
defined on the measure space, $\Omega$, of H\"older-continuous functions on the circle $\T$. In \eqref{cov},
$\Delta$ denotes the Laplacian on $\T$.
(We explicitly exhibit $\beta$, but not $m$, in our notation for the Gaussian measure, 
since the choice of the parameter $m>0$ in the reference measure is entirely
inconsequential. It is not a parameter in the NLS equation, and in expressions such as (\ref{canX}) and (\ref{microcanX}) it is added in and subtracted out again for convenience, namely to avoid fussing about zero modes.)

As explained in the introduction, since the value of $N(\phi)$
is conserved under the NLS flow, it is natural to consider the 
probability measures $\nu_{n,\beta}$ obtained by
conditioning $\mu_\beta$ on the subsets $\{N(\phi) = n\}, n \in \mathbb{R}_{+}$. The construction of these measures involves conditioning on a subset of the support of $\mu_{\beta}$ of measure zero. We must therefore establish regularity properties of the functions $\phi$ in the support of $\mu_{\beta}$ that allow us to use $\nu_{n,\beta}$ as a reference measure in the construction of a canonical Gibbs measure.

In the next subsection we derive integrability estimates for certain exponential functionals when intergrated  against $\mu_\beta$  and  $\nu_{n,\beta}$. These estimates will be used in our construction
of Gibbs measures.

\subsection{Exponential integrability with respect to $\mu_\beta$ and $\nu_{\beta,n}$}

For $k\in \Z$, we set $u_k(x): = L^{-1/2}e^{2\pi i kx/L}$. We note that $\{u_k\}_{k\in\Z}$ 
is a an orthonormal basis in $L^2(\T)$. 
We regard the complex-valued random variables
$$X_k(\phi) = \langle u_k,\phi\rangle$$
as $\R^2$- valued random variables under the obvious identification of $\mathbb{C}$ with $\mathbb{R}^{2}$. Evidently, they are mutually independent with respect to $\mu_\beta$.   
Define $\theta_k$ by
$$
\theta_k := \langle u_k,C u_k\rangle^{-1}  = \beta((2\pi k/L)^2 + m^2)\ .
$$
Then the law of $X_k(\phi)$ on the plane has the density
$$\frac{\theta_k}{2\pi}\exp\left(- \frac12 \theta_k(x_1^2+ x_2^2)\right)\ $$
with respect to Lebesgue measure on the plane. 

It follows that each $\|X_k(\phi)\|_2^2$ is $\Gamma$- distributed,
with density given by
$$
\frac{\theta_k s}{2}\exp\left(- \frac12 \theta_k s\right)\
 $$
with respect to Lebesgue measure $\dd s$ on $\R_+$.
Since each summand in  $\sum_{k\in \Z}\|X_k(\phi)\|^2$ has a smooth, strictly positive density, and since the convolution of such a measure with any probability measure
has a smooth and strictly positive density, it follows that $N(\phi) $
has a smooth density $\rho_{N}(s)$ that is strictly positive for all $s > 0$. 

Next, let 
$$
Y_k(\phi) : =   \langle u_k,[\beta (m^2-\Delta)]^{\gamma/2}\phi\rangle,
$$
with $\gamma > 0$, which we regard as an $\R^2$- valued random variable. Of course,
$$Y_k(\phi) = \theta_k^{\gamma/2} X_k(\phi) =  \theta_k^{\gamma/2}  X_k(\phi) \ ,$$
for any $k \in \mathbb{Z}$. Therefore, for $\xi,\epsilon\in \R$,
\begin{eqnarray}
\int _\Omega \exp(i\xi \|X_k\|_2^2+ \epsilon \|Y_k\|_2^2)\dd \mu_\beta &=& 
\int _\Omega \exp([ i\xi  + \epsilon \theta_k^\gamma]\|X_k\|_2^2)\dd \mu_\beta\nonumber\\
&=& \frac{\theta_k}{  \theta_k  -2i\xi -2   \epsilon \theta_k^\gamma}\nonumber\\
&=& \left( 1 - \frac{2i\xi  }{\theta_k}  -\frac{2\epsilon  }{\theta_k^{1-\gamma}} \right)^{-1}\ .\nonumber
\end{eqnarray}

Since $\theta_k = \mathcal{O}(k^2)$, it is evident that as long as $\gamma < 1/2$ and $2\epsilon \leq (\beta m^2)^{1-\gamma}$,
$$F(\xi,\epsilon) = \prod_{k\in \Z} \left( 1 - \frac{2i\xi  }{\theta_k}  -\frac{2\epsilon  }{\theta_k^{1-\gamma}} \right)^{-1}$$
converges, and hence that
$$
 \int e^{i\xi \|\phi\|_2^2}e^{\epsilon\|(m^2-\Delta)^\gamma\phi\|_2^2}{\rm d}\mu_{\beta}\ = F(\xi,\epsilon)
 $$
 converges.

We use the fact that $F(\xi,\epsilon)$ decays rapidly in $\xi$, in order to establish regularity properties of the
 measures $\nu_{n,\beta}$. 

\begin{lm}\label{rapid}  For all $\gamma < 1/2$ and $2\epsilon \leq (\beta m^2)^{1-\gamma}$,
there are finite positive constants $K_1$ and $K_2$ depending only on $\gamma$, $\epsilon$, $\beta$, $L$ and $m$ so that 
\begin{equation}\label{decay}
\Re \log F(\xi,\epsilon)  \leq K_1  -  K_2 \sqrt{(|\xi |-1)_+}\ .
\end{equation}
\end{lm}

\begin{proof} 
For any $j\in \mathbb{Z}$, we replace $\theta_j$ by an arbitrary complex variable $z_j$. We then define
 $f(\xi,\epsilon) := \log F(\xi,\epsilon)|_{\theta_j \rightarrow z_j}$, i.e.,
$$f(\xi,\epsilon) = - \sum_{j\in 2\pi L^{-1}\Z} \ln \left[ 1 -2\frac{i\xi + \epsilon z_j^\gamma}{z_j}\right]\ .$$
Since $\gamma < 1/2$, the series is summable. For real $z_j$,
\begin{eqnarray}
-\Re f(\xi)  &=&
\frac12 \sum_{j\in 2\pi L^{-1}\Z} \ln  \left[ (1 -2\epsilon z_j^{\gamma-1})^2 + 4\frac{ 
 \xi^2}{z_j^2}\right]\nonumber\\
&=&  \sum_{j\in 2\pi L^{-1}\Z} \ln (1 - 2\epsilon z_j^{\gamma-1})
+\frac12 \sum_{j\in 2\pi L^{-1}\Z} \ln  \left[ 1  + 4\frac{ 
 \xi^2}{z_j^2 (1 -2\epsilon z_j^{\gamma-1})^2} \right]
\nonumber
\end{eqnarray}
As long as $2\epsilon \leq (\beta m^2)^{1-\gamma}$, $1 - 2\epsilon z_j^{\gamma-1}> 0$. Then each logarithm in the sum on the right is positive, 
and
$$
\ln  \left[ 1  + 4\frac{ 
 \xi^2}{z_j^2 (1 -2\epsilon z_j^{\gamma-1})^2} \right]
  \geq \ln  \left[ 1  + 4\frac{ 
 \xi^2}{z_j^2 } \right]
\ .$$
Therefore, 
\begin{eqnarray}
\sum_{j\in 2\pi L^{-1}\Z} \ln  \left[ 1  + 4\frac{ 
 \xi^2}{z_j^2 (1 -2\epsilon z_j^{\gamma-1})^2} \right]
&\geq &
\sum_{j\in 2\pi L^{-1}\Z} \ln  \left[ 1  + 4\frac{ \xi^2}{z_j^2 } \right]\nonumber\\
&\geq& 
\sum_{j\in L^{-1}\Z, z_j \leq |\xi|} \ln  \left[ 1  + 4\frac{ \xi^2}{z_j^2 } \right]\nonumber\\
&\geq& K_2\sqrt{(|\xi |-1)_+}\ .\nonumber
\end{eqnarray}

The final estimate is valid since there are $\mathcal{O}(\sqrt{(|\xi |-1)_+})$ terms in the final sum, and each is at least as large as $\ln 5$. 

Since the remaining sum is convergent, there is a constant $K_1$ depending only on the indicated quantities, in  particular,  independent of $\xi$, so that (\ref{decay}) is valid. 
\end{proof}
We are now ready to construct the measures $\nu_{n,\beta}$ and to study their support. We recall that, for $\ell\in \N$, 
 the {\em Fejer kernel} $K_\ell(t)$ is defined on $[-\pi,\pi]$ by
 $$K_\ell(t) = \frac{1}{\ell}\sum_{k=0}^{\ell-1}\left(\sum_{j=-k}^k e^{ijt}\right) = \frac{1}{\ell}\frac{ 1- \cos(\ell t)}{1-\cos t}\ ,$$ 
 see, e.g. \cite{Z}.
We extend $K_\ell(t)$ to all of $\R$ by setting $K_\ell(t) = 0$, for $|t|> \pi$. 
  As is well known, $K_\ell(t) \geq 0$, for all $t$, for every $\ell$, $\int_{-\pi}^\pi K_\ell(t)\dd t = 1$, 
 and, for all $\delta > 0$, 
 $$
 \lim_{\ell \to\infty} \int_{[-\pi,-\delta]\cup[\delta,\pi]}K_\ell(t)\dd t =0 \ .
 $$
  For $n>0$, $\ell\in \N$ and $\beta>0$, we then define the probability measure $\nu_{n,\beta,\ell}$ by
\begin{equation} \label{numeas}
\nu_{n,\beta,\ell} =   
         \left(\int_\Omega K_\ell(n - N(\phi))\dd \mu_\beta\right)^{-1} K_\ell(n - N(\phi))\dd \mu_\beta\ .
\end{equation}
By what we have noted above,
$$\lim_{\ell\to\infty}  \int_\Omega K_\ell(n - N(\phi))\dd \mu_\beta  = \rho_N(n)$$
where $\rho_N$ is the strictly positive density of $N(\phi)$.   Next,
$$ \left(\int_\Omega K_\ell(n - N(\phi))\dd \mu_\beta\right) \int_\Omega  e^{\epsilon\|(m^2-\Delta)^\gamma\phi\|_2^2}{\rm d}\nu_{n,\beta,\ell} = 
\frac{1}{\ell}\sum_{k=0}^{\ell-1}\left(\sum_{j=-\ell}^\ell e^{-ijn}F(j,\epsilon)\right)\ .$$
The rapid decay of $F(\xi,\epsilon)$ in $\xi$, for fixed $\epsilon$, exhibited in Lemma~\ref{rapid} shows that 
there is a constant $K< \infty$ independent of $\ell$ and $n$ such that 
$$ \int_\Omega  e^{\epsilon\|(m^2-\Delta)^\gamma\phi\|_2^2}{\rm d}\nu_{n,\beta,\ell} \leq K\ .
$$

H\"older norms are useful at several places in our work.
Let $\psi$ be a continuous  real-valued periodic function on $[0,L]$. For $0 < \alpha < 1$, $\psi$ is H\"older-continuous of order $\alpha$ if
$$\sup_{0\leq x<y \leq L}\frac{|\psi(y) - \psi(x)|}{(y-x)^\alpha} =   \|\psi\|_\ha < \infty\ .$$
Then $\|\psi\|_\ha$ is the order-$\alpha$ H\"older semi-norm of $\psi$. 

Uniformly  bounded sets in $\Omega:= \mathcal{C}(\T)$ on which some H\"older norm 
is bounded are compact, by the Arzela-Ascoli Theorem, and 
thus one way to prove that a family of Borel probability measures on $\Omega$ is {\em tight}, so that a limit point exists,  is 
to establish uniform bounds on the rate at which  the probability of $\{\phi\ :\ \|\phi\|_\ha > t\}$ tends 
to zero, as $t$ tends to $\infty$.
The bounds on integrals of exponentials of fractional Sobolev norms 
proven above readily imply tightness  via
the following standard Sobolev Embedding Theorem 
for functions on the one-dimensional torus:

\begin{lm}[Sobolev Embedding] For all $0 < \alpha < 2\gamma - 1/2$, there is a universal constant $C_{\alpha,\gamma}$  such that, for all 
functions $\phi$ on the torus in the domain of the operator $(-\Delta)^{\gamma}$,
$$
\|\phi\|_\ha \leq C_{\alpha,\gamma} 
\|(-\Delta)^\gamma \phi\|_2 \ .
$$
\end{lm}

This shows that, for all $\alpha < 1/2$, there is an $\epsilon>0$ and a constant $K < \infty$ such that 
$$
 \int_\Omega  e^{\epsilon\|\phi\|_\ha^2}{\rm d}\nu_{n,\beta,\ell} \leq K\ ,
$$
uniformly in $\ell$. Thus, the family of probability measure $\{\nu_{n,\beta,\ell}\}_{\ell\in \N}$ is tight. A
simple computation shows that, for any $h \in L^2(\T)$,
$\lim_{\ell\to\infty} \int_\Omega e^{\langle h,\phi\rangle} {\rm d}\nu_{n,\beta,\ell}$
exists, and hence the sequence $\{\nu_{n,\beta,\ell}\}_{\ell\in \N}$ has a unique limit point. 
We may therefore define a probability measure $\nu_{n,\beta}$ by
$$
\nu_{n,\beta} = \lim_{\ell\to\infty}  \nu_{n,\beta,\ell}\ .
$$
By construction, 
\begin{equation}\label{disintegration}
\mu_{\beta} = \int_0^\infty \rho_N(n) \nu_{n,\beta} \dd n
\end{equation}
is the disintegration of $\mu_\beta$ obtained by conditioning on subsets on which $N(\phi)$ has a fixed value.  We have proven the following lemma.

\begin{lm} Let $\nu_{n,\beta}$ be the law of $\phi$ under  $\mu_{\beta}$ conditional on $N(\phi)= n$. 
Let $0 \leq \gamma < 1/2$. 
Then there is an $\epsilon> 0$  and a constant $K< \infty$,  
depending only $\epsilon$  and $\gamma$, such that 
$$ \int_\Omega  e^{\epsilon\|(m^2-\Delta)^\gamma\phi\|_2^2}{\rm d}\nu_{n,\beta} \leq K\ .
$$
\end{lm}

\subsection{Construction of the canonical Gibbs measures}

As mentioned in the introduction, the canonical Gibbs measures constructed and studied in this paper are given, formally, by
\begin{equation}\label{gibbs}
\frac{1}{Z} \delta(N(\phi) - n) e^{-\beta H(\phi)} \mathcal{D} \phi  \mathcal{D} \overline{\phi}\ ,
\end{equation}
and making sense of this expression boils down to showing integrability of 
$e^{(\beta \lambda /p) \|\phi\|_p^p}$ with respect to the  reference measure  $\nu_{n,\beta}$. Here we present a proof of this property for all $\lambda > 0$, assuming that $p < 6$.

The  integrability of $e^{(\beta \lambda /p) \|\phi\|_p^p}$, $\lambda >0$,  was first studied in \cite{LRS}. 
In this paper, a {\em generalized canonical ensemble} has been considered in which the $\delta$-function in (\ref{nunbf})
 is replaced by the characteristic function of the subset $\{\phi\ :\ N(\phi) \leq n\}$ of $\Omega$, and the normalization is adjusted appropriately. It is shown in \cite{LRS} that, for $p< 6$,
integrability holds for all $\lambda>0$, and that, for $p=6$, integrability holds for all $\lambda < \Lambda_0$, for some $\Lambda_0> 0$.  There is however a gap in the proof for the case in which $\phi$ is periodic with period $L$, i.e., $\phi$ is defined on $\T$. The proof in \cite{LRS} is complete for the case where $\phi$ is defined on the interval $[0,L]$
with boundary condition $\phi(0)= 0$ (or any other constant) , and with $\phi(L)$ free. The proof described here covers the case of periodic boundary conditions, but
does not apply when $p=6$, no matter how small $\lambda$ is chosen.

The next lemma bounds the sup-norm, $\|\psi\|_\infty$, of a function $\psi$ by a geometric mean of 
the $L^{2}$-norm $\|\psi\|_2$ and the H\"older semi-norm $\|\psi\|_\ha$ in which the power of $\|\psi\|_\ha$ is greater than $1/2$,
but may be brought arbitrarily close to $1/2$ by choosing $\alpha$ sufficiently close to $1/2$.

\begin{lm}\label{linf} For all $\psi$ with $ \|\psi\|_\ha < \infty$,
$$\|\psi\|_\infty \leq   \frac{\|\psi\|_2}{\sqrt{L}} + 2\|\psi\|_2^{2\alpha/(2\alpha+1)} \|\psi\|_\ha^{1/(2\alpha+1)} \ .$$
\end{lm}  

\begin{proof} By Chebychev's inequality, the set
$E_M := \{ x :  |\psi(x)| \geq M\}$
has Lebesgue measure no greater than $\|\psi\|_2^2/M^2$.  In particular, for $M > \|\psi\|_2/\sqrt{L}$, $E_M^c\neq \emptyset$,
and every point in $[0,L]$  lies within a distance of  $\|\psi\|_2^2/M^2$ of a point in $E_M^c$, or else $E_M$ would contain an interval of length greater than $\|\psi\|_2^2/M^2$, which is impossible.
Fix an $M > \|\psi\|_2/\sqrt{L}$. 

Next, choose $y$ such that 
$|\psi(y)| = \|\psi \|_\infty$ and then choose $x\in E_M^c$ such that $|x-y| \leq  \|\psi\|_2^2/M^2$. By the definition of the H\"older norm,
$$|\psi(y) - \psi(x)| \leq \|\psi\|_\ha |y-x|^\alpha \leq   \|\psi\|_\ha \left(\frac{\|\psi\|_2^2}{M^2}\right)^\alpha\ ,$$
and hence
$$|\psi(y)| \leq M +   \|\psi\|_\ha \left(\frac{\|\psi\|_2^2}{M^2}\right)^\alpha\ .$$
If $M> 0$ could be chosen freely we would choose it to be a multiple of $(\|\psi\|_2^{2\alpha}\|\psi\|_\ha)^{1/(2\alpha +1)}$.  Since we require
$M  > \|\psi\|_2/\sqrt{L}$, we choose
$$M = \frac{\|\psi\|_2}{\sqrt{L}} +  (\|\psi\|_2^{2\alpha}\|\psi\|_\ha)^{1/(2\alpha +1)}\ ,$$
and this yields the bound stated in the lemma. 
\end{proof} 

\begin{thm}[Existence of the canonical  Gibbs measure]
There exist finite constants $C_1$, $C_2$ and $C_3$ such that, for all $2 <  p < 6$, 
and for all $\lambda$  and all $n> 0$,
\begin{equation}\label{bnd2A}
\int_\Omega   e^{\frac{\lambda}{p} \int \|\phi\|_p^p } {\rm d}\nu_{n,\beta} < 
C_3 \exp\left(C_1\lambda n^{p/2} +  C_2\lambda^{(4\alpha p+2)/(4\alpha+4-p)}n^{(2\alpha p+2)/(4\alpha+4-p)}\right)\ .
\end{equation} 
\end{thm}

\begin{proof}
Lemma~\ref{linf}  yields the bound
 $$ 
 \|\phi\|_p^p \leq \|\phi\|_2^2 \|\phi\|_\infty^{p-2} \leq   \left(\frac{2}{\sqrt{L}}\right)^{p-2} \|\phi\|_2^p +  4^{p-2}
 \|\phi\|_2^{2(\alpha p+1)/(2\alpha +1)}  \|\phi\|_\ha^{(p-2)/(2\alpha +1)}\ .
 $$
 As long as $(p-2)/(2\alpha +1) < 2$, we can apply our bound on the integrability of the exponential of
 $\|\phi\|_{\ha}^{2}$. 
For any $p< 6$, we can choose $\alpha < 1/2$ such that
 this is the case.  We then apply  the arithmetic-geometric mean inequality 
$$
ab \leq \frac{1}{s'} (a/\epsilon)^{s'} + \frac{1}{s}(\epsilon b)^s,
$$
where $1/s'+1/s =1$ and $s,s'>1$,
with 
$a = \|\phi\|_2^{(2\alpha p+2)/(2\alpha +1)}$ and $b = \|\phi\|_\ha^{(p-2)/(2\alpha +1)}$. Choosing
$$
s= \frac{4\alpha+2}{p-2}, \qquad{\rm so \ that} \qquad s' = \frac{4\alpha+2}{4\alpha +4 - p} \ ,$$
we obtain
\begin{eqnarray}
\|\phi\|_2^{(2\alpha p+2)/(2\alpha +1)}  \|\phi\|_\ha^{(p-2)/(2\alpha +1)} &\leq&
\epsilon^{-(4\alpha+2)/(4\alpha+4-p)}\frac{4\alpha +4 - p}{4\alpha+2} \|\phi\|_2^{(4\alpha p+4)/(4\alpha+4-p)}\nonumber\\
&+&   \epsilon ^\frac{4\alpha+2}{p-2} \frac{4\alpha +4 - p}{4\alpha+2}  \|\phi\|_\ha^2\ .
\end{eqnarray}
By choosing $\epsilon$ sufficiently small we therefore arrive at a bound of the form
$$
\lambda \|\phi\|_p^p \leq C_1\lambda \|\phi\|_2^p  +  C_2\lambda^{(4\alpha p+2)/(4\alpha+4-p)}\|\phi\|_2^{(4\alpha p+4)/(4\alpha+4-p)} + \delta \|\phi\|_\ha^2\ ,
$$
for any $\delta>0$. 
Choosing $\delta$ so small that the exponential of $\delta\|\phi\|_\ha^2$ is integrable, we conclude that 
$$
\int_\Omega   e^{\frac{\lambda}{p} \int \|\phi\|_p^p } {\rm d}\nu_{n,\beta} \leq  \exp\left(C_1\lambda n^{p/2} +  C_2\lambda^{(4\alpha p+2)/(4\alpha+4-p)}n^{(2\alpha p+2)/(4\alpha+4-p)}\right)\int_\Omega   e^{\delta \|\phi\|_\ha^2 }  {\rm d}\nu_{n,\beta} \ .
$$
\end{proof}

We now turn to studying the generalized grand-canonical ensemble. Our main result is the following theorem the first part of which guarantees the existence of the generalized grand-canonical ensemble. 
The second part of the theorem will be used in our proof of the fact that the stochastic dynamics 
relaxes to a Gibbs measure exponentially fast.

\begin{thm}\label{intthm}  For all $2 <  p < 6$, 
and for all $\lambda$ and all $\kappa> 0$,
\begin{equation}\label{bnd2}
\int  e^{\frac{\lambda}{p} \|\phi\|_p^p -  \kappa \|\phi\|_2^{2r}} {\rm d}\mu_\beta < \infty\ ,
\end{equation} 
whenever 
\begin{equation}\label{mm}
r > \frac{2\alpha p +2}{4\alpha +4-p}.
\end{equation}
Furthermore,  for all $\mu \in \mathbb{R}$ and all $\kappa> 0$, and for all $\gamma < 1/8$, 
\begin{equation}\label{sgb2}
\int \exp\left[ \mu \left|\int_{\T} |\phi|^2 \overline {\phi } (m^2-\Delta)^\gamma \phi {\rm d}x\right|   - \kappa\|\phi\|_2^{2r}\right] {\rm d}\mu_\beta  < \infty
\end{equation}
whenever $r > 9$. 
\end{thm}

\begin{proof} Because of (\ref{disintegration}) and (\ref{bnd2A}),  we have (\ref{bnd2}) whenever (\ref{mm}) is satisfied.

We now turn to the quantity in (\ref{sgb2}), which we bound by
$$ \left|\int_{\T} |\phi|^2 \overline {\phi } (m^2-\Delta)^\gamma \phi {\rm d}x\right|  \leq   \|\phi\|_\infty^2 \|\phi\|_2 \|(m^2-\Delta)^\gamma\phi\|_2 \ .$$
Noticing that
 $$
 \|(m^2-\Delta)^\gamma \phi\|_2^2 = \langle \phi, (m^2-\Delta)^{2\gamma} \phi\rangle \leq \|\phi\|_2 \|(m^2-\Delta)^{2\gamma}\phi\|_2\ ,
 $$
we find that
\begin{equation}\label{swa}
 \left|\int_{\T} |\phi|^2 \overline {\phi } (m^2-\Delta)^\gamma \phi {\rm d}x\right|  \leq   \|\phi\|_\infty^2 \|\phi\|_2^{3/2} \|(m^2-\Delta)^{2\gamma}\phi\|_2^{1/2}  \ .
 \end{equation}

Now, as long as 
$4\gamma < \tfrac12$, $\|(m^2-\Delta)^{2\gamma}\phi\|_2 < \infty$, almost surely with respect to 
$\mu_\beta$.
There is therefore a $\delta' > 0$ such that
$$
\int e^{\delta'  \|(m^2-\Delta)^{2\gamma}\phi\|_2^2}{\rm d}\mu_\beta < \infty\ .
$$
We now apply the arithmetic-geometric mean inequality in the form
$$
abc \leq \frac{1}{s_1}(\epsilon a)^{s_1} +  \frac{1}{s_2}(\epsilon^{-2}b)^{s_2} + \frac{1}{s_3}(\epsilon c)^{s_3}\ ,
$$
where
$$\frac{1}{s_1} + \frac{1}{s_2}+ \frac{1}{s_3} =1\ .$$
The proof is completed by choosing $s_1 = 3/2$, $s_2 = 12$ and $s_3 = 4$.
\end{proof}

We note, in passing, that the techniques presented in this section can be used to give proofs of various 
well known regularity properties, such as H\"older continuity, of Wiener paths.

\section{Absence of ``phase transitions''}

In this section we study analyticity properties of the partiton functions of the measures constructed in the last section. 

The first thing to notice is that
the normalization factor (partition function) needed to normalize the reference measure $\mu_\beta$, i.e., the Gaussian measure with
mean $0$ and covariance $\beta^{-1}(m^2-\Delta)^{-1}$, is {\em independent of $\beta$}, up to a `constant factor' formally given by `$\left(\det[ \beta I]\right)^{1/2}$', which can be set to $1$ by a suitable redefinition of the formal Liouville measure $\mathcal{D}\phi\mathcal{D}\overline{\phi}$. 
This means that we may take the partition function in the generalized grand-canonical Gibbs measure to be given by
\begin{equation}\label{ggcanp}
Z(\kappa,\beta,\lambda) := \int_\Omega \exp\left[\beta\left(\frac{\lambda}{p}\|\phi\|_p^p + \frac{m^2}{2}\|\phi\|_2^2 - \frac{\kappa}{r}\|\phi\|_2^{2r}\right)\right]{\rm d}\mu_\beta\ ,
\end{equation}
replacing the formal expression involving an integral with respect to $\mathcal{D}\phi\mathcal{D}\overline{\phi}$, because there is no non-trivial
$\beta$- dependence in the constant that normalizes $\mu_\beta$. 

In order to determine the dependence of $Z$  on the parameters $\kappa$, $\beta$ and $\lambda$, we change variables in the functional integral \eqref{ggcanp}, setting
\begin{equation}\label{varch}
\sqrt{\beta}\phi = \psi\ .
\end{equation}

The image of the Gaussian $\mu_\beta$ under this change of variables is $\mu_1$, and hence
\begin{equation}\label{ggcanp2}
Z(\kappa,\beta,\lambda) = \int_\Omega \exp\left[\frac{\lambda \beta^{1-p/2}}{p}\|\phi\|_p^p + \frac{m^2}{2}\|\phi\|_2^2 - \frac{\beta^{1-r}\kappa}{r}\|\phi\|_2^{2r}\right]{\rm d}\mu_1\ .
\end{equation}
As long as 
\begin{equation}\label{domain}
\Re \beta^{1-r}\kappa > 0\ , \qquad \text{with} \text{   } \lambda \text{   } \text{arbitrary},
\end{equation}
our exponential integrability results show that $Z(\kappa,\beta,\lambda)$ is well defined.  Thus, Eq. (\ref{domain}) specifies a large domain of analyticity of the
partition function in the variables $\kappa, \beta$ and $\lambda$. It follows that the ``generalized pressure'', $\log( Z(\kappa,\beta,\lambda) )$, is analytic in its arguments on the subset of this domain where
$Z(\kappa,\beta,\lambda) \neq 0$. Since, for $\beta > 0$,  $\kappa>0$ and $\lambda$ real, 
$Z(\kappa,\beta,\lambda)\neq 0$, we obtain the following result.

\begin{thm}
The generalized pressure, $\log (Z(\kappa,\beta,\lambda))$, is analytic in $\kappa, \beta$ and $\lambda$ on a domain including the set of all $\kappa, \beta$ and $\lambda$ such that 
$\Re \beta$ and $\Re \kappa$ are strictly positive, and $|\Im \beta|$, $|\Im \kappa|$ and $|\Im \lambda|$ are sufficiently small (depending on $\Re \kappa$, $\Re \beta$ and $\Re \lambda$). 
\end{thm}

Similar arguments can be used to prove analyticity properties of the ``generalized free energy'' in the canonical ensemble. We define the canonical partition function, $\tilde{Z}$, by 
\begin{equation}
\label{cpf}
\widetilde Z(n,\beta,\lambda) := e^{\frac{\beta}{2}m^{2}n} \int_{\Omega} \exp \beta\left[ \frac{\lambda}{p} \|\phi\|_p^p\right] {\rm d}\nu_{n,\beta}\ .
\end{equation}
Let $n = \beta^{-1} k$, with $\beta, k >0$. 
In the functional integral on the right side of \eqref{cpf} we change variables:
\begin{equation}
\label{varch}
 \psi:= \sqrt{\beta} \phi.
 \end{equation}
In the $\psi$-variables, the measure $\nu_{n,\beta}$ becomes $\nu_{k,1}$. Thus,
\begin{equation}
\label{cpf'}
\widetilde Z(n=\beta^{-1}k,\beta,\lambda) = e^{\frac{1}{2}m^{2}k} \int_{\Omega} \exp\left[\beta^{1- p/2} \frac{\lambda}{p} \Vert \psi \Vert^{p} \right] {\rm d} \nu_{k,1}
\end{equation}
The right side of this equation shows that $\widetilde{Z}(\beta^{-1}k,\beta,\lambda)$ is analytic in the variables $\beta \neq 0$ and $\lambda$. For $k<\infty$, $\beta > 0$ and $\lambda \in \mathbb{R}$, 
$\widetilde{Z}(\beta^{-1}k, \beta, \lambda)$ is strictly positive. Thus, the generalized free energy,
$\log( \widetilde{Z}(\beta^{-1}k,\beta,\lambda))$, is analytic in $\beta$ and $\lambda$ on a domain containing
all $\beta$ and $\lambda$, with $\Re\beta >0$ and $|\Im \beta|$, $|\Im \lambda|$ sufficiently small (depending on $\Re\beta$ and $\Re \lambda$).
 
 This result is similar to one proven in \cite{LRS} for the ``generalized canonical measure''.
(Further analyticity properties can be proven by also rescaling the $x$-variables and changing the diameter of $\mathbb{T}$.)

\section{The stochastic flow}

\subsection{The stochastic differential equation and its generator} 

For the purposes of this section it is useful to write the NLS equation in differential form:
\begin{equation}\label{de}
{\rm d}\phi = \{H, \phi\}{\rm d}t = J D H(\phi) {\rm d } t \ .
\end{equation}
As explained in the introduction, we wish to add a noise term on the right side of this equation with the property that a Gibbs measure is the unique invariant measure for the corresponding stochastic flow. 
Let $\sigma$ be a positive-definite operator on $\K$, and let $b$ be a vector field on $\K$. We
consider the stochastic differential equation (SDE)
\begin{equation}\label{sde}
{\rm d}\phi =   b(\phi){\rm d} t  + \sigma {\rm d}w\ ,
\end{equation}
where $\sigma$ is ``small'' (in a suitable sense), $b(\phi) = JDH(\phi) + O(\sigma^{2})$
is a small perturbation of the Hamiltonian vector field $JDH(\phi)$, and ${\rm d}w$ denotes a Wiener process on $\K$.  
We seek to choose $b$ such that a Gibbs measure is invariant under the stochastic flow
solving this SDE, and then we will attempt to control the
rate of relaxation to the  Gibbs measure for this process.

We intend to review some formulae that are direct analogs of counterparts 
in a finite-dimensional variant of this problem. 
It may clarify what we have to require of $b$ and of $\sigma$ and what we will need to estimate if we first investigate a finite-dimensional model problem. 
We do this by using dimension-independent methods that extend easily to the 
infinite-dimensional setting, so that very little of this analysis will need to be 
repeated.  Moreover, some of the results for the finite-dimensional case appear to be new 
and are of independent interest. 

\subsection{A finite-dimensional model: Invariant measures and drift vector fields} 

Let $C$ be a real, positive-definite $2n \times 2n$ matrix, and let $\mu_0$ be the Gaussian probability measure on $\mathbb{R}^{2n}$ given by
\begin{equation} \label{Gaussian}
{\rm d}\mu_0 = \frac{1}{Z_C} e^{-x\cdot (2C)^{-1}x}{\rm d}x\ .
\end{equation}
Let $V$ be a smooth function on $\R^{2n}$ that is bounded from below and adjusted such that 
\begin{equation}
\label{Gibbsmeas}
{\rm d}\mu_V := e^{-V}{\rm d}\mu_0
\end{equation}
is a probability measure. 
We define a function $H$ on $\mathbb{R}^{2n}$ by
\begin{equation}\label{fdH}
H := V+ x\cdot (2C)^{-1}x - \ln Z_C,
\end{equation}
so that 
$${\rm d}\mu_V  = e^{-H} {\rm d}^{2n}x\ .$$
  
Let $b(x) $ be a (bounded Lipschitz-continuous) vector field on $\R^{2n}$. We consider the SDE 
\begin{equation}\label{BSDE}
dx_t = b(x_t){\rm d}t + \sigma {\rm d}w_t\ ,
\end{equation} 
where $\sigma$ is a positive-definite matrix on $\mathbb{R}^{2n}$ and $dw$ is Brownian motion on 
$\mathbb{R}^{2n}$.
The question of interest in this section is: {\em For which choices of the drift vector field $b$
is $\mu_V$ an invariant measure for the SDE in Eq. (\ref{BSDE})?}
Under mild conditions on $b$, the process determined by \eqref{BSDE} conserves probability; 
we suppose this to be the case, for now, and will verify it later when we make a specific choice 
for the drift vector field $b$.

By Ito's formula, 
$$
\lim_{h\downarrow 0} \frac{1}{h}  \mathbb{E} (\varphi(x_{t+h}) - \varphi(x_t) \ | \ x_t) = \mathcal{L}\varphi (x_t),$$
for any smooth function $\varphi$ on $\R^{2n}$, where
$$
 \mathcal{L}\varphi(x) := \frac12 \Delta_{\sigma^2}\varphi(x) + b(x)\cdot \nabla \varphi(x), 
 \qquad{\rm with}\qquad   \Delta_{\sigma^2}\varphi(x)  := \nabla \cdot \sigma^2 \nabla \varphi(x)\ .
 $$
 
 Suppose that $x_t$ is a solution of the SDE (\ref{BSDE}) and that the initial distribution, i.e., the law of $x_0$,  is $f_0\mu_V$. Then the distribution of 
 $x_t$ is of the form $f_t\mu_V$, where $f_t$ is a smooth function.   Then 
 $$
 \frac{{\rm d}}{{\rm d} t} \mathbb{E}\varphi(x_t) =   \int_{\R^{2n}} \mathcal{L} \varphi(x)  f_t(x)\dd \mu_V =
 \int_{\R^{2n}}  \varphi(x) \frac{\partial}{\partial t}f_t(x)\dd \mu_V\ .
 $$
 It follows that $f_t$ satisfies the equation
  \begin{equation}\label{invar0}
 \frac{\partial}{\partial t}f_t(x) = \mathcal{L}^*f_t(x),
 \end{equation}
 where $\mathcal{L}^*$ is the adjoint of $\mathcal{L}$ in $L^2(\mu_V)$. 
 
 In particular,  it follows that $\mu_V$ is an invariant measure for the process solving the SDE \eqref{BSDE}
 if and only if 
 \begin{equation}\label{invar}
 \int_{\R^{2n}} (\mathcal{L}\varphi)\dd \mu_V= 0,
 \end{equation}
 for all $\varphi$, or, equivalently, if $\mathcal{L}^* 1 = 0$.  
 Integrating by parts one finds that
 $$
 \int_{\R^{2n}} (\Delta_{\sigma^2}f)e^{-H}{\rm d}^{2n} x =  \int_{\R^{2n}} f \left[ \nabla H \cdot \sigma^2 \nabla H - \Delta_{\sigma^2}H\right] e^{-H}\ ,
 $$
 and
 $$\int_{\R^{2n}} (b\cdot \nabla f)e^{-H}{\rm d}^{2n} x  = \int_{\R^{2n}} f \left[  \nabla H \cdot b  -\nabla \cdot b  \right] e^{-H}{\rm d}^{2n} x \ .$$
 Thus, (\ref{invar}) is satisfied if and only if 
 $$ \frac12 \left[ \nabla H \cdot \sigma^2 \nabla H - \Delta_{\sigma^2}H\right]  + \left[  \nabla H \cdot b -\nabla \cdot b \right]  =0\ .$$
This holds true whenever
  \begin{equation}\label{invar2}
   b = u +v, 
 \end{equation}
with
 \begin{equation}\label{invar3}
   u = - \frac12 \sigma^2 \nabla H \qquad  {\rm and} \qquad   \nabla H \cdot v -\nabla \cdot v = 0\ .
    \end{equation}
    
If $J$ is the usual symplectic matrix on $\R^{2n}$ and $v = J\nabla H$ 
then $\nabla H \cdot v =0$ and $\nabla \cdot v =0$. We therefore take
 \begin{equation}\label{invar4}
   b  = - \frac12 \sigma^2 \nabla H + J \nabla H\ .
    \end{equation}
    
To the extent that $\sigma$ is  ``small'' in an appropriate sense, the SDE (\ref{BSDE}), with $b$ given by (\ref{invar4}), is a ``small'' stochastic perturbation
of the deterministic Hamiltonian flow
$$\dd x_t = J\nabla H(x_t)\dd t\ ,$$
and $\mu_V$ is  an invariant measure for the process described by this SDE. 
In the next subsection we will discuss the relaxation of initial distributions to the invariant Gibbs measure. 

Note that the set-up considered here is not the same as the one for the standard Ornstein-Uhlenbeck process on phase space, which corresponds to a Langevin equation with velocity variables 
that enter the Hamiltonian quadratically, and in which the
noise acts only on the velocity variables and not on the position variables. 
Here the noise acts on all of the canonical variables, and this facilitates exponential relaxation. 

\subsection{Exponential relaxation in the finite-dimensional model}

 Suppose that $x_t$ is a solution of the SDE (\ref{BSDE}) and that the initial distribution, i.e., the law of $x_0$,  is given by a distribution $f_0\mu_V$. Then the distribution of 
 $x_t$ is of the form $f_t\mu_V$, where $f_t$ is smooth, 
 and we have seen that $f_t$ satisfies (\ref{invar0}). 
 In this subsection we further suppose 
 that $f_0 \in L^2(\mu_V)$.  
  Then
 $ \frac{\partial}{\partial t}(f_t(x)-1) = \mathcal{L}^*f_t =   \mathcal{L}^*(f_t -1),$
 so that
\begin{equation}\label{gapA}
 \frac{{\rm d}}{{\rm d}t} \| f_t - 1\|^2_{L^2(\mu_V)} = 2\langle (f_t -1) \mathcal{L}^*(f_t -1)\rangle_{L^2(\mu_V)}  = 
 \langle (f_t -1)[\mathcal{L} +  \mathcal{L}^*](f_t -1)\rangle_{L^2(\mu_V)}\ .
 \end{equation}
This suggests to define the operator
\begin{equation}\label{Hdef}
\H = -\frac12 [\mathcal{L} + \mathcal{L}^*]\ .
\end{equation}
We note that this operator is self-adjoint on $L^2(\mu_V)$. 

For any smooth test function $f$, 
$$\langle f, \H f \rangle_{L^2(\mu_V)}  =  \frac12\int_{R^{2n}} \nabla f \cdot \sigma^2 \nabla f {\rm d}\mu_V\ .$$
We define the {\em ``interacting'' Dirichlet form}, $\E(f,f)$, by
\begin{equation} \label{Dirichlet}
\E(f,f) =  \frac12\int_{R^{2n}} \nabla f \cdot \sigma^2 \nabla f {\rm d}\mu_V\ ,
\end{equation}
and we note that $\H$ is the non-negative operator associated to this quadratic form. 

It is clear that $\E(f,f) =0$ if and only if $f$ is constant, and so the null space of $\H$ is spanned by the constant function $1$. 
The spectral gap of $\H$  is  the quantity defined by
\begin{equation}\label{gap}
E_1 := \inf\{  \langle f, \H f\rangle_{L^2(\mu_V)}\ :    \langle f, 1\rangle_{L^2(\mu_V)}= 0 \quad {\rm and}\quad   \langle f, f\rangle_{L^2(\mu_V)}= 1\} \ .
\end{equation}
It then follows from (\ref{gapA}) and (\ref{Hdef}) that
\begin{equation}\label{gapB}
 \frac{{\rm d}}{{\rm d}t} \| f_t - 1\|^2_{L^2(\mu_V)}   = 
 \langle (f_t -1)2\mathcal{H}(f_t -1)\rangle_{L^2(\mu_V)} \leq -2E_1  \| f_t - 1\|^2_{L^2(\mu_V)} \ ,
 \end{equation}
 so that 
 \begin{equation}\label{gapC}
  \| f_t - 1\|_{L^2(\mu_V)}   \leq e^{-tE_1}   \| f_0 - 1\|_{L^2(\mu_V)} \ .
 \end{equation}
 Thus, when $f_0\in L^2(\mu_V)$, $f_t\mu_V$ relaxes exponentially fast to the equilibrium distribution $\mu_V$.

The part of $\mathcal{L}$ involving $v$, the Hamiltonian part of the drift field, does not contribute to the Dirichlet form $\mathcal{E}$, and hence it does not affect the exponential relaxation rate $E_1$. Of course, there is exponential relaxation only if $E_1>0$.
In the finite-dimensional case, 
there are many ways to show, under our contidions,  that $E_1$ is non-zero, and that it is an eigenvalue of $\H$. 

In the infinite-dimensional case, 
 the positivity of $E_1$ is not an obvious property and does not hold in general.
However, if we can show that $e^{-t\mathcal{H}}$ is compact or, better yet, trace-class then we conclude,
once again, that $E_1$ is the second smallest eigenvalue of 
$\H$ and is strictly positive. 

Therefore, it is of interest to estimate $\tr[e^{-t\H}]$ by a method that can be applied in the infinite-dimensional situation. We conclude our discussion of the finite-dimensional case by explaining such a method enabling us to estimate  $\tr[e^{-t\H}]$.  The estimate on this quantity that we obtain appears to be new even in the finite-dimensional case.  Although we quantitatively estimate  $\tr[e^{-t\H}]$, we do not find a quantitative estimate on $E_1$.  Under the conditions considered here,
which include the example of a function $H(x)$, see (\ref{fdH}), to correspond to a double well potential, the gap $E_1$ can be exponentially small 
in the parameter playing the role of $\lambda$.  

To carry out this program,  we introduce the {\em free Dirichlet form} $\E_0(f,f)$ through
\begin{equation}\label{freedirdef}
\E_0(f,f) = \int_{R^{2n}} \nabla f \cdot \sigma^2 \nabla f \dd \mu_0\ .
\end{equation}
We define a non-negative operator $\H_0$ on $L^2(\mu_0)$ by
\begin{equation}\label{freeHdef}
\langle f, \H_0f\rangle_{L^2(\mu_0)} = \E_0(f,f) \ .
\end{equation}

The operator $\H_0$ is a very familiar, simple operator; 
at least in the case where $\sigma$ is a power of $C$. 
Let us suppose that $\{u_1,\dots,u_{2n}\}$ is an orthonormal basis of $\R^{2n}$ consisting of common eigenvectors of $C$ and $\sigma$. Let
$\nu_k$ denote the $k$th eigenvalue of $C$; i.e., $Cu_k = \nu_ku_k$. We suppose that there is an $s$, with
$0<s<1/2$, such that 
$$\sigma u_k = \nu_k^s u_k, \qquad \text{for} \text{  } \text{all} \text{   }k.$$
Then, when expressed in terms of partial derivatives in coordinates derived from the orthonormal 
basis  $\{u_1,\dots,u_{2n}\}$,
$$
\H_0 f =  \sum_{k=1}^{2n}\nu_k^{2s}  \left( -\frac{\partial^2}{\partial x_k^2} + \frac{x_k}{\nu_k}  \frac{\partial}{\partial x_k} \right).
$$
The operators  ${\displaystyle  \left( -\frac{\partial^2}{\partial x_k^2} + \frac{x_k}{\nu_k}  \frac{\partial}{\partial x_k} \right)}, k=1, ..., 2n,$ all commute with one another and 
are what one knows as scaled ``number operators'', which are diagonalized by Hermite polynomials. 
Their eigenvalues are given by $m\nu_k^{-1}$, with $m=0,1,2,...,$ and each eigenvalue has multiplicity one.
We thus have the following lemma.

\begin{lm}\label{traceclass}
$$\tr[e^{-t\H_0}] = \prod_{k=1}^{2n} \frac{1}{1 - e^{-t\nu_k^{2s-1}}}\ .$$
\end{lm} 
It follows that
$$
\log\left( \tr[e^{-t\H_0}] \right) = -\sum_{k=1}^{2n} \log\left( 1 - e^{-t\nu_k^{2s-1}}\right),
$$
and this is readily estimated in terms of the sum
$$
\sum_{k=1}^{2n} e^{-t\nu_k^{2s-1}}\ ,
$$
which converges, as $n$ tends to infinity, under mild growth conditions on $\nu_{k}^{-1}$, assuming that $s<1/2$.

Next, we show that $\H$ is unitarily equivalent to an operator of the form $\H_0 + U$, where $U$ is a multiplication operator with the property that the negative part of 
$\tau U$ is exponentially integrable with respect to $\mu_0$, for all $\tau > 0$.   The next lemma explains the relevance of this fact.

\begin{lm}\label{hyper}
$$\tr[e^{-2t(\H_0+U)}] \leq \tr[e^{-t\H_0}] \left(\int_{\R^{2n}}e^{-2C_{{\rm LS}}U}\dd \mu_0\right)^{ t/C_{{\rm LS}}}$$
where $C_{{\rm LS}}$ is defined to be the largest of the numbers $2\nu_k^{1-2s}$, with $1-2s>0$.

\end{lm}

\begin{proof}  The Golden-Thompson Inequality implies that, for any bounded continuous function $U$,  
$$\tr [e^{-2t(\H_0+U)}] \leq \tr[e^{-t\H_0} e^{-t(H_0 + 2U)}] \leq  \tr[e^{-t\H_0} \| e^{-t(H_0 + 2U)}]\|_\infty $$
where  $\| \cdot \|_\infty$ denotes the operator norm, and 
$$\|e^{-t(H_0 + 2U)}]\|_\infty = e^{t\lambda(2u)}$$
where $-\lambda(2u)$ is the bottom of the spectrum of $\H_0+ 2U$:
\begin{equation}\label{Feder1}
\lambda(2U) = \sup \left\{-2\langle \varphi, U \varphi \rangle_{L^2(\mu_0)} -   \langle \varphi, \H_0 \varphi\rangle_{L^2(\mu_0)}\ :\  \|\varphi\|_2 =1    \right\}
\end{equation}

Because $\H_0$ satisfies a logarithmic Sobolev inequality, we can bound $\lambda(2U)$ in terms of $\int e^{-2U}\dd \mu_0$, using an argument of Federbush \cite{Federbush} that
we now recall (with some further optimization).
 For all $b\in \R$ and $r,a>0$,  
 $$ab \leq  ra(\log a + \log r) + \frac{1}{e}e^{b/r}\ .$$
 Therefore, for all $r>0$,  and $\|\varphi\|_{L^2(\mu_0) }=1$, 
\begin{equation}\label{Feder2}
-2\langle \varphi, U \varphi \rangle_{L^2(\mu_0)} \equiv  \int_{\R^{2n} }(-2U) |\varphi|^2 \leq r \int_{\R^{2n}} | \varphi|^2 \log  | \varphi|^2 \dd \mu_0 + r\log r  + \frac{1}{e} \int_{\R^{2n}} e^{-2U/r}\dd\mu_0\ .
\end{equation}
 
 The Logarithmic Sobolev Inequality satisfied by $\H_0$ says that for $\|\varphi\|_{L^2(\mu_0) }=1$,
 \begin{equation}\label{tens0}
\int_{\R^{2n}} | \varphi|^2 \log  | \varphi|^2 \dd \mu_0    \leq C_{{\rm LS}} \langle \varphi, \H_0 \varphi \rangle_{L^2(\mu_0)} .
\end{equation}
where $C_{{\rm LS}}$ is defined to be the largest of the numbers $2\nu_k^{1-2s}$. 

Combining (\ref{Feder1}) and (\ref{Feder2}) with $r = 1/C_{{\rm LS}}$, we obtain
\begin{equation}\label{Feder3}
\lambda(2U) \leq \frac{1}{C_{{\rm LS}}}\log \frac{1}{C_{{\rm LS}}} + \frac{1}{e} \int_{\R^{2n}}e^{-2C_{{\rm LS}}U}\dd \mu_0\ .
\end{equation}
Now observe that if we replace $U$ by $U+ u/2$ where $u\in \R$, we have 
$$
\lambda(2U)  =  \lambda(2U + u) + u \leq  \frac{1}{C_{{\rm LS}}}\log \frac{1}{C_{{\rm LS}}} + \frac{e^{-uC_{{\rm LS}}}}{e} \int_{\R^{2n}}e^{-2UC_{{\rm LS}}}\dd \mu_0 + u\ .$$
Optimizing over $u$, we obtain
\begin{equation}\label{ImFe}
\lambda(2U)   \leq    \frac{1}{C_{{\rm LS}}}  \log \left( \int_{\R^{2n}}e^{-2C_{{\rm LS}}U}\dd \mu_0\right)
\end{equation}
Therefore,
$$e^{t\lambda(2U)} \leq  \left(\int_{\R^{2n}}e^{-2C_{{\rm LS}}U}\dd \mu_0\right)^{ t/C_{{\rm LS}}}\ .$$
\end{proof}

\if false

It follows that 
$$\lambda(2U) \leq   \int_{\R^{2n}} e^{4[U]_-}\dd \mu_0\ $$.
Therefore, $e^{-t(H_0 + 2U)}$ is a bounded operator. We already know that $e^{-t\H_0}$ is trace class, and the product of
a bounded operator and a trace class operator is trace class.

For each $k$. let $A_k$ denote the operator
$$A_k = \nu_k^{2s}  \left( -\frac{\partial^2}{\partial x_k^2} + \frac{x_k}{\nu_k}  \frac{\partial}{\partial x_k} \right)$$
so that $\H_0 = \sum_{k=1}^{2n}A_k$. Let $\gamma_k$ denote the Gaussian probability measure
$$\gamma_k(x) = \frac{1}{\sqrt{2\pi \nu_k}}e^{-x^2/2\nu_k}\ .$$
Integrating by parts, for any test function $\varphi$ on $\R$,
$$\int_\R \varphi A_k \varphi \dd \gamma_k = \nu_k^{2s} \int_\R |\nabla \varphi|^2 \dd \gamma_k\ .$$
The Logarithmic Sobolev Inequality states that 
$$
\int_\R | \varphi|^2 \log  | \varphi|^2 \dd \gamma_k  - \left(\int_\R | \varphi|^2 \dd \gamma_k\right)\log 
 \left(\int_\R | \varphi|^2 \dd \gamma_k\right)  \leq \nu_k^{2s-1}   \int_\R |\nabla \varphi|^2 \dd \gamma_k\ .$$
 Let $C_{\rm LS}$ be given by
 $$
C_{\rm LS}  = \inf_{k} \nu_k^{2s-1}\ .
$$
Then by the tensorization properties of Logarithmic Sobolev inequalities,  for any test function $\varphi$ on $\R^{2n}$,
\begin{equation}\label{tens}
\int_{\R^{2n}} | \varphi|^2 \log  | \varphi|^2 \dd \mu_0  - \left(\int_{\R^{2n}} | \varphi|^2 \dd \mu_0\right)\log 
 \left(\int_{\R^{2n}} | \varphi|^2 \dd \mu_0\right)   \leq C_{\rm LS} \langle \varphi, \H_0 \varphi \rangle_{L^2(\mu_0)} .
\end{equation}

Let $W$ be the operator on $L^2(\mu_0)$ of multiplication by a real-valued bounded continuous function $W$. Let 
$-\lambda(W)$ denote the bottom of the spectrum of $\H_0 + W$. The the dual (by Legendre transform) of
the Logarithmic Sobolev inequality (\ref{tens}) is the inequality
$$\lambda(W) \leq \int_{\R^{2n}} e^{2[W]_-}\dd \mu_0\ $$.

\end{proof}

  Let $x_t$ denote a solution of our SDE with initial distribution $f_0\mu_V$ where $f_0\in L^2(\mu_V)$. 
        For any smoooth function test function $\varphi$, 
        $$
        \frac{{\rm d}}{{\rm d}t} \mathbb{E}\varphi(x_t) = \mathbb{E}(\mathcal{L}\varphi(x_t))\ .
        $$
$$\lim_{h\downarrow 0} \mathbb{E}(f(x_t)f(x_{t+h}) - f^2(x_t))  = \int_{\R^{2n}} f \mathcal{L}f {\rm d}\mu_V = -\mathcal{E}(f,f)  = -\frac12 \int_{R^{2n}} \nabla f \cdot \sigma^2 \nabla f {\rm d}\mu_V\ .$$

We define the {\em Dirichlet form} $\E(f,f)$ by
\begin{equation}\label{dirinta}
\E(f,f) = \frac12 \int_{R^{2n}} \nabla f \cdot \sigma^2 \nabla f {\rm d}\mu_V
\end{equation}
and the operator $H$ by 
\begin{equation}\label{dirintb}
 \int_{R^{2n}} f \H f  {\rm d}\mu_V  =   \E(f,f)\ .
\end{equation}
Integrating by parts one finds that
 \begin{equation}\label{invar4}
 \mathcal{H}  
 = -\frac12 \Delta_{\sigma^2} + \frac12 \nabla H \sigma^2\cdot \nabla\ .
 \end{equation}
The part of $\mathcal{L}$ involving $v$, the Hamiltonian part of the drift field, makes no contribution to the Dirichlet
form $\mathcal{E}$, and hence not to the operator $\H$. 

Note that $\H = -\tfrac12\left(\mathcal{L} + \mathcal{L}^*\right)$
The adjoint of $\mathcal{L}$ has been computed in $L^2(\mu_V)$, and evidently $\H$ is self adjoint on that space. Indeed, the
operator $\H$ is evidently positive semidefinite, and $\E(f,f)= 0$ if and only if $f$ is constant. The the null space of $\H$ is one dimensional and is spanned by
the constant function $1$.

Our goal in this subsection is to   show that $\H$ has a strictly positive spectral gap, and to do this using an argument that can be adapted to the infinite
dimension setting that is our main interest. 
Since we know that $0$ is an eigenvalue of $\H$ of multiplicity one, and hence that for all $t>0$, $1$ is an eigenvalue of
$e^{-t\H}$ of multiplicity one, it suffices to observe that in the finite dimensional case $1$ cannot be an accumulation point of the spectrum since
there are only finitely many eigenvalues.

When we turn to the infinite dimensional analog, we will encounter an operator $\H$ for which $0$ will be an eigenvalue of multiplicity one. 
If we can then show that $e^{-t\H}$ is trace class, and in particular compact,  for all $t>0$, we shall know once again that $1$ cannot be an accumulation point of
the spectrum of $e^{-t\H}$, and thus that $0$ cannot be an  accumulation point of
the spectrum of $\H$. This will prove the existence of a spectral gap. 

Suppose that we know that $E_1>0$. Consider any initial distribution for our stochastic process of the form $f\mu_V$ with 
$f\in L^2(\mu_V)$. Then $f -1$ is orthogonal
to $1$ in $L^2(\mu_V)$, and so 
$$\|e^{-t\mathcal{L}}f  -1\|_{L^2(\mu_V)}  \|e^{-t\H}f  -1\|_{L^2(\mu_V)} = \|e^{-t\H}(f  -1)\|_{L^2(\mu_V)}  \leq e^{-tE_1}\|f-1\|_{L^2(\mu_V)}\ .$$
Since $[e^{-t\mathcal{L}}f]\mu_V$ is the distribution of the solution $x_t$ of our SDE at time $t$, we have the exponential relation that we seek.

We shall now explain a method for obtaining a bound on $\tr e^{-t\H}$ that will extend to the infinite dimensional setting. While it is obvious in
the finite dimensional setting that $e^{-t\H}$ is trace class, this is not so evident in the infinite dimensional setting. 

Next, note that  for all $t>0$, $e^{-t\H}$ is a contraction on $L^2(\mu_V)$, and that $e^{-t\H}1=1$. Since it is easily seen that the 
Peron-Frobenius Theorem applies $e^{-t\H}$, we see that $1$ is an eigenvalue of multiplicity one. It follows immediately that $\H$ has a spectral gap.
That is
$$E_1 := \inf\{  \langle f, \H f\rangle_{L^2(\mu_V)}\ :    \langle f, 1\rangle_{L^2(\mu_V)}= 0 \quad {\rm and}\quad   \langle f, f\rangle_{L^2(\mu_V)}= 1\} > 0\ .$$
Therefore, 
$$\|e^{-t\H}f -1\|_{L^2(\mu_V)} \leq  e^{-tE_1}\|f-1\|_{L^2(\mu_V)}\ .$$
Since if the initial distribution of the solution of our SDE is $f\mu_V$, the distribution at time $t$ will be $(E^{-t\H}f)\mu_V$, this proves that the distribution 
relaxes to the equilibrium distribution exponentially fast. 

There are two main ingredients to this argument. The first is a bound on $\tr e^{-t\H}$ and the second is a Perron-Frobenius argument. Both
steps will require some care in the infinite dimensional case that we discuss in full detail in the next subsections.  While the fact that  $\tr e^{-t\H} < \infty$ is trivial in the finite dimensional model,
it is worthwhile to obtain an explicit bound by a method that may be adapted to the infinite dimensional case.  We close this subsection by proving such a bound.

\fi

We conclude this subsection by proving that  $\H$ is unitarily  equivalent to $\H_0+U$ for some potential 
$U$ such that $e^{-\tau U}$ is integrable with respect to $\mu_0$, for all $\tau>0$. 

Let $\mu_V := e^{-V} \mu_0$. For arbitrary $f\in L^2(\mu_V)$, we define
\begin{equation} \label{mu_V}
 Tf = e^{-V/2}f. 
 \end{equation}
 Then $T$ is unitary from $L^2(\mu_V)$ onto $L^2(\mu_0)$. 
We fix a smooth $f\in L^2(\mu_V)$. 

\begin{lm}\label{utran} Let $T$ be the unitary operator defined in (\ref{mu_V}). Then 
$$T\H T^{-1} = \H_0 + U,$$
where
$$U = \frac14 \nabla V \cdot \sigma^2 \nabla V + \frac12 \H_0 V\ .$$

\end{lm}

\begin{proof}  To simplify our notation, we put $g := Tf = e^{-V/2}f$.   
Then, by the definition of $T$, 
$$
\nabla g = \nabla f e^{-V/2} - \frac12 g \nabla V,
$$
so that
$$\nabla f = \left(\nabla g + \frac12 g \nabla V\right)e^{V/2}\ .$$
Therefore
\begin{eqnarray}
\mathcal{E}(f,f) &=&  \int_{\R^{2n}}\left( \nabla g + \frac12 g \nabla V\right)\cdot \sigma^2 \left( \nabla g + \frac12 g \nabla V\right) {\rm d}\mu_0\nonumber\\
&=&  \int_{\R^{2n}} \nabla g \cdot \sigma^2 \nabla g {\rm d}\mu_0  +  \frac14 \int_{\R^{2n}} g^2 (\nabla V \cdot\sigma^2 \nabla V){\rm d}\mu_0\nonumber\\
&+&  \int_{\R^{2n}} \nabla g \cdot g\sigma^2 \nabla V {\rm d}\mu_0\ .\label{ham1}
\end{eqnarray}
Of course, 
$$
\int_{\R^{2n}} \nabla g \cdot \sigma^2 \nabla g {\rm d}\mu_0 = \langle g, \H_0 g\rangle_{L^2(\mu_0)},
$$
 and the last term can be simplified:
$$ 
\int_{\R^{2n}} \nabla g \cdot g\sigma^2 \nabla V {\rm d}\mu_0 = \frac12  \int_{\R^{2n}} \nabla g^2 \cdot \sigma^2 \nabla V {\rm d}\mu_0
=\frac12 \langle g^2, \H_0 V\rangle_{L^2(\mu_0)}\ .
$$
Altogether,
$$\langle f,\H f\rangle_{L^2(\mu_V)} \equiv  \mathcal{E}(f,f)  =  \langle g, \H_0 g\rangle_{L^2(\mu_0)}  +   \frac14 \int_{\R^{2n}} g^2 (\nabla V \cdot\sigma^2 \nabla V){\rm d}\mu_0
+ \frac12 \langle g^2, \H_0 V\rangle_{L^2(\mu_0)}\ .$$

\end{proof}

To render the formula for the effective potential  $U$ more transparent, we need an explicit expression 
for $\H_0V$. 
Such an expression can be found by using the integration by parts formula on Gauss space:
$$
 \int_{\R^{2n}} \nabla g^2 \cdot \sigma^2 \nabla V {\rm d}\mu_0  = 
\int_{R^{2n}} g^2 \left[ (-\nabla \cdot \sigma^2 \nabla V) + \nabla V \cdot \sigma^2 C^{-1}x\right]{\rm d}\mu_0\ ,
$$
where $C$ is the covariance of $\mu_0$; see (\ref{Gaussian}).
Thus, 
\begin{equation}\label{H0Vform}
\H_0V(x) = -\nabla \cdot \sigma^2 \nabla V +  \sigma^2 C^{-1}x \cdot \nabla V \ .
\end{equation}

As stated in Lemma 4.2, in order to bound $\tr (e^{-t\H})$, it suffices to bound 
$\int_{\R^{2n}} e^{-sU}\dd \mu_0$, for sufficiently large $s$. In our example, this will turn out to be 
finite, for all $s>0$. 

At this point we must make an explicit choice for $V$. 
The function $V$ that would arise in a finite-dimensional approximation to the original, infinite-dimensional problem (considered in the next section) 
has two terms.  One of these is a multiple of $|x|^{2r}$. 
This motivates us to consider the example
$$
V(x) = |x|^{2r}, 
$$
for $r > 2$. We shall see that the desired exponential integrability holds true if and only if 
$\sigma^2$ is neither ``too large'' nor ``too small''. 

For the present choice of $V$, $\nabla V = 2r|x|^{2r-2}x$, hence
\begin{equation}\label{term1}
\nabla V \cdot \sigma^2 \nabla V = 4r^2 |x|^{4r-4}(x\cdot \sigma^2 x)\ ,
\end{equation}

\begin{equation}\label{term2}
\nabla V \cdot \sigma^2 C^{-1}x   = 2r|x|^{2r-2}(x\cdot \sigma^2C^{-1} x)\ ,
\end{equation}
and
\begin{equation}\label{term3}
\nabla  \cdot \sigma^2 \nabla V   =  2r|x|^{2r-2}\tr(\sigma^2) + 4r(r-1)|x|^{2r-4}(x\cdot \sigma^2 x)\ ,
\end{equation}

In our proof of exponential relaxation to equilibrium, we will require the negative part of $U$ to be exponentially integrable.  The problematic contribution
to the potential $U$ is the first term in  (\ref{term3}).  First of all, since this includes a factor of $\tr (\sigma^2)$, we will have to assume that $\sigma^2$ be trace-class 
when we extend our analysis to infinite-dimensional examples.  However, even if this assumption is taken for granted, that term is still problematic, because  it involves $|x|^2 $, rather than $x\cdot \sigma x$,
which can be very small in high dimension if $\sigma^2$ is trace class. 
Luckily, the term (\ref{term2}) comes to our rescue, as we may choose $\sigma^2C^{-1} > I$. With such a choice, $U$ turns out to be bounded below (for the present choice of $V$),
\textit{independently of dimension}, and hence exponentially integrable. 

The choice $V(x) = |x|^{2r}$  models the term $\|\phi\|_2^{2r}$ in the modified Hamiltonian of the grand-canonical ensemble. Of course, we also need to take into account the main term, namely
$\tfrac{\lambda}{p}\|\phi\|_p^p$. We shall discuss this only in the infinite-dimensional model to which the next subsection is devoted. 
The observations made, so far, are not misleading; they suggest the right assumptions to be made in the study of the infinite-dimensional model: We shall require
$\sigma^2$ to be \textit{trace-class, but not too small}, in order to be able to conclude exponential relaxation
to a Gibbs state.

\subsection{The infinite-dimensional SDE}

We restrict our attention to a self-interaction term proportional to 
$$
- \int_{\T}\vert \phi(x) \vert^{4} dx,
$$ 
i.e., an exponent $p=4$, and discuss more general interaction terms later. Thus, we define
\begin{equation}\label{Vint}
V(\phi) = -\frac{\lambda}{4}\|\phi\|_4^4 + \kappa \|\phi\|_2^{2r} - \frac{1}{\beta} \log \widetilde Z(\kappa,\beta,\lambda)\ ,
\end{equation}
and note that
\begin{equation}\label{normalized}
\int_\Omega e^{-\beta V(\phi)}\dd \mu_\beta = 1\ ,
\end{equation}
see Theorem \ref{intthm}. 
The generalized grand-canonical Gibbs measure is denoted by 
$\widetilde \gamma_{\kappa,\beta,\lambda}$, and we then have that
$$\text{d}\widetilde \gamma_{\kappa,\beta,\lambda}  = \frac{1}{\widetilde Z(\kappa,\beta,\lambda)} e^{-\beta V(\phi)} 
\text{d}\mu_\beta\ .$$

Next, we introduce two Dirichlet forms.

\begin{defi} [Free and Interacting Dirichlet Forms]  Let $\sigma$ be a positive Hilbert-Schmidt operator. A {\em ``free'' Dirichlet form}, $\mathcal{E}_0$,
is defined by 
\begin{equation}\label{freedir}
\mathcal{E}_0(f,f) = \int_\Omega \langle Df(\phi), \sigma^2 Df(\phi)\rangle_\K \dd \mu_\beta(\phi) \ ,
\end{equation}
where $D$ is the Fr\'echet derivative defined in (\ref{Fder}).
The {\em ``interacting'' Dirichlet form} $\mathcal{E}$ is  defined by 
\begin{equation}\label{freedir}
\mathcal{E}(f,f) = \int_\Omega \langle Df(\phi), \sigma^2 Df(\phi)\rangle_\K \dd \widetilde \gamma_{\kappa,\beta,\lambda}(\phi) \ .
\end{equation}
\end{defi}

Since both Dirichlet forms are closeable on a natural domain of smooth functions, they determine two self-adjoint operators, $\mathcal{H}_0$  and $\mathcal{H}$,
by 
\begin{equation}\label{Lfree}
  \int_\Omega  f (\mathcal{H}_0 f )\dd \mu_\beta =   \mathcal{E}_0(f,f)  \ 
\end{equation}
and 
\begin{equation}\label{Lint}
  \int_\Omega  f (\mathcal{H} f) \dd \widetilde \gamma_{\kappa,\beta,\lambda}  =  \mathcal{E}(f,f)  \ .
\end{equation}

The operator $\mathcal{H}_0$ is a familiar object: it is a direct sum of multiples of number operators labelled by the wave vectors (mode indices) $k\in \mathbb{Z}$. Under mild conditions saying that
$\sigma^2$ is not too small, one verifies that
$e^{-t\mathcal{H}_0}$ is trace-class, for all $t>0$.   
In fact, these conditions are easily read off from Lemma~\ref{traceclass} in the previous subsection. 

It will be convenient to have an explicit form of $\H_0$. In order to avoid uninteresting complications, we assume 
that $\sigma$ is a power of the covariance $C = (-\Delta + m^{2})^{-1}$ of the Gaussian measure 
$\mu_{\beta}$: 
\begin{equation} \label{sigma}
\sigma = C^s,
\end{equation}
for some $s>0$ to be specified below. 
Let $\{u_k\}$ be an orthonormal basis in $\K$ consisting of eigenfunctions of
$C$, and hence of $\sigma^2$.
Of course, the eigenvectors $\{u_k\}$ of $C$ form the usual trigonometric basis, and the index $k$ 
(the wave vector) ranges over $\Z$. The eigenvalues, $\nu_k$, of $C$ corresponding to these eigenfunctions are given by 
$$
\nu_k =  \beta^{-1}((2\pi k/L)^2 + m^2)^{-1}\ .
$$

Let $D_k$ denote the directional derivative in the direction of $u_k$; i.e., for a smooth function $F$ on $\K$,
$$D_kF(\phi)  = \langle u_k, DF(\phi)\rangle_\K\ .$$
Then
\begin{equation}\label{H0Vform2}
\H_0F(x) = -\sum_{k\in \Z}\left[  \nu_k^{2s} D^2_k F(\phi)  -  \nu_k^{2s-1} \langle u_k,\phi\rangle_\K  D_k F(\phi)\right]\ .
\end{equation}

Lemma~\ref{traceclass} generalizes immediately to yield:

\begin{lm}\label{traceclass2} Under the assumption that $s < 1/2$, 
$e^{-t\H_0}$ is trace-class, and 
$$\tr[e^{-t\H_0}] = \prod_{k=1}^{\infty} \frac{1}{1 - e^{-t\nu_k^{2s-1}}}\ .$$
\end{lm}

Our goal is to prove that $e^{-t\H}$ is trace-class, too, for all $t>0$. It is straightforward to adapt the proof of Lemma~\ref{utran} to show that 
$\H$ is unitarily equivalent to $\H_0 + U$, for an explicit multiplication operator $U$ on 
$L^{2}(\Omega, \mu_{\beta})$; (henceforth, we will omit the $``\Omega$'' from our notation). Lemma~\ref{hyper} holds true independently of dimension and applies
to our infinite dimensional problem without any changes.  Thus, our main task, in this subsection, is to determine the explicit form of $U$ and to then prove that
$e^{-rU}$ is integrable with respect to $\mu_\beta$, for all $r>0$.

The first step in carrying out this task is to construct a unitary transformation from 
$L^2(\gamma_{\kappa,\beta,\lambda})$ to $L^2(\mu_\beta)$. 

For $f\in L^2(\gamma_{\kappa,\beta,\lambda})$,  we define 
\begin{equation}\label{Tdef}
Tf = fe^{-\beta V(\phi)/2}\ .
\end{equation}
Because of our normalization of $V$, see (\ref{Vint}), $T$ is unitary from  
$L^2(\gamma_{\kappa,\beta,\lambda})$ to $L^2(\mu_\beta)$.  

\begin{lm}
For an arbitrary smooth function $f\in L^2(\gamma_{\kappa,\beta,\lambda})$, we define $g = Tf \in  L^2(\mu_\beta)$.
Then
\begin{equation}
\label{trans}
\mathcal{E}(f,f)  =  \mathcal{E}_0(g,g)   +  \int_\Omega g^2 U \dd \mu_\beta,
\end{equation}
where the potential $U(\phi)$ is given by
\begin{equation}
\label{trans2}
U(\phi) = \frac{\beta}{2} (\mathcal{H}_0 V)(\phi) + \frac{\beta^2}{4}
 \|\sigma D V(\phi)\|_{L^{2}(\mu_{\beta})}^{2}\ .
\end{equation}
In particular, $\mathcal{H}$ is unitarily equivalent to the operator
$$\mathcal{H}_0 + U\ $$
acting on $L^{2}(\mu_{\beta})$.
\end{lm} 

\begin{proof}
Fix a smooth $f\in L^2(\mu_\beta)$. Then with $g = Tf$,
$$Dg = (Df )e^{-\beta V(\phi)/2}  -\alpha g \ $$
where
$$\alpha(\phi) = \frac{\beta}{2}DV(\phi)\ .$$
It follows that
$$Df = (Dg + \alpha g) e^{\beta V(\phi)/2}$$
and hence that 
\begin{eqnarray}
\mathcal{E}(f,f) &=& \int_\Omega \|\sigma(Dg + \alpha g)  \|^2_{L^{2}(\mu_{\beta})} \dd \mu_\beta\nonumber\\
&=&  \mathcal{E}_0(g,g) + 2 \int_\Omega \langle Dg,\sigma^2\alpha g\rangle_\K  \dd \mu_\beta  +  \int_\Omega g^2  \|\sigma\alpha   \|^2_{L^{2}(\mu_{\beta})} \dd \mu_\beta\ .\nonumber
\end{eqnarray}
We now observe that
$$
2 \int_\Omega \langle Dg,\sigma^2 \alpha g\rangle_\K  \dd \mu_\beta  =   \int_\Omega \langle D(g^2),\sigma^2 \alpha \rangle_\K  \dd \mu_\beta 
= \frac{\beta}{2} \mathcal{E}_0(g^2, V)\ .$$
It then follows from the definition of $\mathcal{H}_0$ that
$$2 \int_\Omega \langle Dg,\sigma^2 \alpha g\rangle_\K  \dd \mu_\beta  =  \int_\Omega g^2 \mathcal{H}_0 V
 \dd \mu_{\beta}\ .$$

\end{proof}

We now make a special choice of $V$, namely
\begin{equation}\label{VIdef}
V(\phi) = -\frac{\lambda}{4}\|\phi\|_4^4 + \frac{\kappa}{2r}\|\phi\|_2^{2r}\ .
\end{equation}

We propose to prove that, with this choice of $V$, and with $U$ defined by  (\ref{trans2}), 
$e^{-sU}$ is integrable with respect to $\mu_\beta$, for any $s>0$,
provided only that $r>9$ and $\sigma^2C^{-1}  \leq (-\Delta + m^2)^{\gamma}$, with $0 <\gamma < 1/8$. 
Since Lemma~\ref{hyper} generalizes directly to infinitely many dimensions, we will then 
have  proven that $e^{-t\H}$ is trace- class. 
Since it is evident from the definition of the Dirichlet form that the null space of $\H$ is spanned by the constant functions, one concludes that the spectral gap $E_1$ is strictly positive. 
Thus, the next result provides the exponential relaxation to a Gibbs state that we 
have been seeking to prove.

\begin{thm}\label{exintthm2}
Let  $V$ be given by (\ref{VIdef}), with $r>9$, and suppose that 
$I \leq \sigma^2C^{-1}  \leq (-\Delta + m^2)^{\gamma}$, 
with $0 <\gamma < 1/8$.
Let  $U$ be  given by (\ref{trans2}). Then

$$
\int_\Omega e^{-\tau U}\dd \mu_\beta < \infty\ ,
$$
for all $\tau>0$.
\end{thm}

\begin{proof}   Because $U(\phi) \geq (\H_{0}V)(\phi)$, it suffices to prove that 
\begin{equation}\label{integr}
\int_\Omega e^{-\tau (\H_{0}V)(\phi)}\dd \mu_{\beta}(\phi) < \infty\ ,
\end{equation}
for all $\tau > 0$.

By direct computation,
$$DV(\phi) = -\lambda |\phi|^2\phi  + \kappa\|\phi\|_2^{2r-2}\phi\ ,$$
and 
\begin{eqnarray}
(\H_{0} V)(\phi) &=& \lambda\left[  |\phi|^2 \tr \sigma^2 +2 \langle \phi ,\sigma^2 \phi\rangle_\K -  \langle \phi ,\sigma^2C^{-1} |\phi|^2 \phi\rangle_\K\right] \nonumber\\
&-& \kappa\left[ \|\phi\|_2^{2r-2}  \tr \sigma^2  +  (2r-2) \|\phi\|_2^{2r-4}\langle \phi ,\sigma^2 \phi\rangle_\K   - 
 \|\phi\|_2^{2r-2}\langle \phi ,\sigma^2C^{-1}  \phi\rangle_\K\right] .
\nonumber
\end{eqnarray}
Thus, choosing the exponent $s$ in (\ref{sigma}) such that $\sigma^2C^{-1} \geq I$, we have that
\begin{eqnarray}
(\H_{0}V)(\phi) &\geq&   -\lambda \langle \phi ,\sigma^2C^{-1} |\phi|^2 \phi\rangle_\K + \frac{\kappa}{2}\|\phi\|_2^{2r}\nonumber\\
&+& \kappa\left[  \frac{1}{2}\|\phi\|_2^{2r}  - (2r-1)\tr \sigma^2 \|\phi\|_2^{2r-2} \right]\ .
\nonumber\\
&\geq& -\lambda \langle \phi ,\sigma^2C^{-1} |\phi|^2 \phi\rangle_\K + \frac{\kappa}{2}\|\phi\|_2^{2r}\nonumber\\
&-&  \frac{2}{2r-2}\left( \frac{(2r-1)(2r-2)}{2r} \tr (\sigma^2) \right)\ .\nonumber
\end{eqnarray}
The results in Sect. 3 then imply that (\ref{integr}) holds.
\end{proof}

\section{Stochastic time evolution for the canonical ensemble}

To obtain a stochastic time evolution that leaves the canonical measure invariant, the noise must preserve the spheres $\{\phi \ :\ N(\phi) = n\}$.
To write down a suitable SDE, let $P_\phi$ denote the orthogonal projection onto the orthgonal complement of $\phi \in \L^2(\T)$. 

Consider the SDE 
$${\rm d}\phi  =    P_\phi \sigma {\rm d}w\  $$
with $\sigma$ as in Section 5.
Then, by Ito's formula,
$${\rm d}\|\phi\|_2^2 = 2 \langle \phi, P_\phi \sigma {\rm d}w \rangle + \frac12 \tr \sigma^2 P_\phi{\rm d}t = \frac12 \tr \sigma^2 P_\phi{\rm d}t\ .$$
Thus, projecting out the normal component of the fluctuations of our noise process still results in a norm that grows, but in a simple way.
We compensate for this by adding a drift: By the same calcluation, the solutions of 
$$ {\rm d}\phi  =     -\frac12 \left(\tr \sigma^2 P_\phi\right) \frac{\phi}{\|\phi\|_2^2}{\rm d}t  +   P_\phi \sigma {\rm d}w\  $$
satisfy ${\rm d}\|\phi\|_2^2 = 0$, so that the sets  $\{\phi \ :\ N(\phi) = n\}$ are invariant.  The same is true for the solutions of
\begin{equation}\label{sdeAA}
{\rm d}\phi =  b(\phi){\rm d}t   -\frac12 \left(\tr \sigma^2 P_\phi\right) \frac{\phi}{\|\phi\|_2^2}{\rm d}t  +   P_\phi \sigma {\rm d}w\  \ ,
\end{equation}
for any vector field $b(\phi)$ that is purely tangential; i.e., such that $P_\phi b(\phi) = 0$, for all $\phi$. 

Of course, $ J D H(\phi) $ is such a vector field, but as in the case of the GGC ensemble, making this choice would not preserve the energy,
for $\sigma \neq 0$. The same computation that provided the appropriate correction in the GGC ensemble shows that
 \begin{equation}\label{invar4A}
   b(\phi)   = - \frac12 \sigma^2 P_\phi \nabla H + J \nabla H(\phi)\ 
    \end{equation}
is the appropriate choice.

Then, proceeding as before, this leads us to consider the Dirichlet form
$$\mathcal{E}_{n,0} = \int |P_\phi \sigma \nabla \phi|^2 {\rm d}\nu_{n,\beta}\ .$$
Unlike the corresponding ``free'' Gaussian Dirichlet form, this form is not a well-studied object, and estimates on its spectrum, in particular ones that would lead to the analogue of Lemma~\ref{traceclass2}, are not available in the literature. Such estimates will be developed in a companion paper where
we will study relaxation to the canonical Gibbs measure. 

\bigskip

\noindent{\bf Acknowledgements} Work of E.A.C. is partially supported
by N.S.F. grant DMS 1201354.  Work of J.L.L. is partially supported by
N.S.F. grant DMR 1104500 and AFOSR grant FA9550.  Most of the work on this
paper was carried out while J.F. and J.L.L. were visiting the School
of Mathematics of the Institute for Advanced Study. We thank our
colleagues at the IAS and, in particular, Tom Spencer H. T. Yau for their
generous hospitality. The stay of J.F. at IAS was supported by
The Fund for Math and The Robert and Luisa Fernholz Visiting
Professorship Fund.  The authors would like to thank Wei-Min Wang for 
discussions during the work.

\vfill\break

\end{document}